\documentclass[12pt]{article}
\usepackage{amsmath,bbm,amssymb,mathrsfs}
\usepackage{amsthm,amscd,float}             
\usepackage{graphicx, multirow, graphics}
\usepackage{comment, color}
\usepackage{longtable, lscape}
\usepackage{rotating}
\usepackage{dcolumn}
\usepackage{natbib}
\usepackage{booktabs,caption,fixltx2e}
\usepackage[flushleft]{threeparttable}
\usepackage{caption}
\allowdisplaybreaks[4]
\newtheorem{theorem}{Theorem}
\newtheorem{lemma}{Lemma}

\numberwithin{equation}{section}

\newcommand{\bm}[1]{\mbox{\boldmath$ #1 $\unboldmath}}

\def\tr{{\rm tr}}
\def\vec{{\rm vec}}
\def \diag{{\rm diag}}

\begin{document}
\baselineskip=22pt
\vskip 20pt

\begin{center}
{\Large \bf On Block Cholesky Decomposition for Sparse Inverse Covariance Estimation}
\end{center}

\begin{center}
Xiaoning Kang, Jiayi Lian and Xinwei Deng$^{*}$
\vskip 5pt {\it
$^{*}$Department of Statistics, Virginia Tech \\
}
\end{center}

\begin{abstract}
The modified Cholesky decomposition is popular for inverse covariance estimation,
but often needs pre-specification on the full information of variable ordering.
In this work, we propose a block Cholesky decomposition (BCD) for estimating inverse covariance matrix under the partial information of variable ordering,
in the sense that the variables can be divided into several groups with available ordering among groups, but variables within each group have no orderings.
The proposed BCD model provides a unified framework for several existing methods including the modified Cholesky decomposition and the Graphical lasso.
By utilizing the partial information on variable ordering, the proposed BCD model guarantees the positive definiteness of the estimated matrix with statistically meaningful interpretation.
Theoretical results are established under regularity conditions.
Simulation and case studies are conducted to evaluate the proposed BCD model.
\end{abstract}


\section{Introduction}
The estimation of covariance and inverse covariance matrices is of fundamental importance in the multivariate statistics with a broad spectrum of applications,
such as linear discriminant analysis \citep{clemmensen2011sparse}, portfolio optimization \citep{deng2013penalized}, and assimilation \citep{nino2019parallel}.
In high-dimensional data, sparse estimation of an inverse covariance matrix has specially attracted great attention,
since it is closely related to a graphical model for inferring the conditional independence between variables of multivariate normal data.
However, estimation of a large inverse covariance matrix often encounters two challenges.
First, the estimated matrix needs to be positive definite for the valid statistical inferences.
Second, the number of parameters in the model increases quickly in a quadratic order in terms of the matrix dimensionality.

Existing studies on the inverse covariance estimation in the literature generally fall into two categories.
Denote the random variables of interest by $\bm X = (X_{1}, \ldots, X_{p})'$ with mean $\bm 0$ for simplicity and covariance matrix $\bm \Sigma$.
The first category needs the pre-specification on the full ordering information of variables $\bm X$.
That is, the variables $X_{1}, \ldots, X_{p}$ have a natural ordering, which typically occurs in longitudinal data, time series, spatial data, spectroscopy and so forth.
In this situation, banded or tapering estimation for the high-dimensional matrices has been developed \citep{bickel2008regularized}, requiring an assumption that the variables are becoming weakly correlated as their positions in the ordering are far away.
Although these methods are straightforward and easy in computation, the resultant estimates may not be positive definite.
A better technique for estimating the inverse covariance matrix in this category is the modified Cholesky decomposition (MCD) introduced by \cite{Pourahmadi1999Joint}.
It not only guarantees the positive definiteness of the estimated matrix,
but also utilizes the information of variable ordering, leading to an accurate estimate.
Moreover, this decomposition has a meaningful regression interpretation, allowing the use of the regularization for sparse estimation \citep{huang2006covariance, kang2021ensemble}.
That is, the sparse pattern in the Cholesky-based matrix estimates can be induced by the sparsity in the Cholesky factors through linear regressions in the decomposition.

The second category considers that the information of variable ordering is not available.
In this case, one strategy is to identify a proper variable ordering based on a certain data driven mechanism,
transforming it into the estimation problem of the first category.
For example, \cite{Wagaman2009Discovering} determined the variable ordering by the Isomap algorithm and proposed an Isoband matrix estimate.
\cite{Dellaportas2012Cholesky} suggested using BIC criterion to seek for the variable ordering before applying the MCD technique.
\cite{Rajaratnam2013Best} recovered the variable ordering via the best permutation algorithm, which is very efficient for autoregressive model.
However, a potential drawback is obvious that the accuracy of such estimates relies on the accuracy of the estimated orderings of variables.
Furthermore, some real data, for example gene data and medical data, practically may not have a natural variable ordering,
which means that it is not adequate to find a variable ordering for such data.
Alternatively, another strategy in this category is to consider a permutation invariant estimation for inverse covariance matrix.
One popular method is the Graphical lasso (Glasso) proposed by \cite{meinshausen2006high} and \cite{yuan2007model}, of which the algorithms and properties have been widely studied \citep{friedman2008sparse, lam2009sparsistency, yuan2010high}.
Other works on the permutation invariant estimation for inverse covariance matrix can be found in \cite{xue2012regularized, wang2015shrinkage, cai2016joint, van2019generalized, wang2020ultrahigh}, among many others.

By comparing two categories of estimation methods, it is seen that the MCD is an appropriate method for estimating the inverse covariance matrix when the information of variable ordering is fully available,
while the Glasso is suitable when the information of variable ordering is not available.
However, in many applications, there is only partial information of variable ordering available.
Here the partial information means that the variables can be divided into several groups (including one group) with the group ordering known but the variable ordering within each group unknown.
For example, in the multi-stage manufacturing process \citep{shi2006stream}, each stage contains a set of variables and the stages have a natural ordering among themselves because of the nature of manufacturing process.
Thus, the full information of variable ordering is not available since the variables within the same stage may not have an ordering.
In this case the partial information of variable ordering is present since there is an ordering among stages.
To estimate the inverse covariance matrix of variables in such a multi-stage manufacturing process, neither the MCD method nor the Glasso method is adequate.
Another concrete example is the Covid-19 data in Case Studies Section.
The data were weekly collected from 37 continuous weeks at the beginning of the pandemic.
In each week, four variables which have no ordering among themselves are recorded.
Such four variables in a certain week are correlated with variables in the former weeks (see details in Case Studies).
Therefore the $4 \times 37 = 148$ variables can be naturally divided into 37 groups by calendar week, which forms a partial information of the variable ordination.

In this work, we fill in the gap to develop a block Cholesky decomposition (BCD) method for estimating the inverse covariance matrix given the partial information of variable ordering.
The proposed BCD method takes advantage of such partial information on the variable ordination to estimate the inverse covariance matrix via the joint estimation of a set of penalized multivariate regressions.
It guarantees the positive definiteness of the estimated sparse inverse covariance matrix with statistically meaningful interpretation.
Moreover, the proposed method provides a unified framework of estimating the inverse covariance matrix, where the MCD, the Glasso method, the \cite{Witten2011new}'s estimator and \cite{Rothman2010A}'s estimator can be all considered as special cases of the proposed method.
The theoretical results suggest that the proposed model can have faster consistent rate than that of Glasso method when the partial information on the variable ordination is present.
The R codes of implementing the proposed model are available at https://github.com/xiaoningmike/BCD.

The remainder of this work is organized as follows.
Section \ref{sec:prop} develops the proposed method along with its parameter estimation.
The asymptotically theoretical property is established in Section \ref{sec:theory} under regularity conditions.
Sections \ref{sec:simulation} and \ref{sec:app} demonstrate the merits of the proposed model via simulations and two real data examples.
We conclude our work with some discussion in Section \ref{sec:con}.
All technical proofs are reported in the Appendix.

\section{The Proposed Method}\label{sec:prop}
In this section, we describe the proposed methodology of the BCD method. Different from the MCD technique, the proposed BCD method only requires partial information of the variable ordination.

\subsection{Block Cholesky Decomposition for Inverse Covariance Matrix}\label{sec:bmcd}
Suppose that the variables in $\bm X$ can be partitioned into $M$ groups, with the $j$th group of variables denoted as $\bm X^{(j)} = (X^{(j)}_{1}, X^{(j)}_{2}, \ldots, X^{(j)}_{p_{j}})', j = 1, 2, \ldots, M$, where $p_{j}$ is the number of variables in the $j$th group, and $\sum_{j=1}^M p_j = p$.
Assume that these $M$ groups of variables have a natural ordering of $\bm X^{(1)}, \bm X^{(2)}, \ldots, \bm X^{(M)}$,
while there is not any ordering structure among variables $X^{(j)}_{1}, X^{(j)}_{2}, \ldots, X^{(j)}_{p_{j}}$ within each group $\bm X^{(j)}$.
We call such ordering information of these $M$ groups as the {\it partial information of the variable ordination}.
Note that when $M=1$, it reduces to the case that there is no ordering information for the $p$ variables $X_{1}, \ldots, X_{p}$,
and $M = p$ corresponds to the case that the $p$ variables have a full ordering information.
In this work, we assume both of $p$ and $M$ can be diverged as the sample size goes to infinity.

Without loss of generality, we write $\bm X = ( (\bm X^{(1)})', (\bm X^{(2)})', \ldots, (\bm X^{(M)})')'$ with its inverse covariance matrix $\bm \Omega = \bm \Sigma^{-1}$.
The key idea of the BCD model is to decompose $\bm \Omega$ by
a block lower triangular matrix constructed from the multivariate regression coefficients when the variable group $\bm X^{(j)}$ is regressed on its preceding variable groups
$\bm X^{(1)}, \bm X^{(2)}, \ldots, \bm X^{(j-1)}$ for $j=2, 3, \ldots, M$.
Specifically, the BCD method considers a series of multivariate regressions
\begin{align}\label{inverse covariance:eq1}
\bm X^{(j)} = \sum_{i=1}^{j -1} \bm A_{ji} \bm X^{(i)} + \bm \epsilon_{j}
       = \bm A_{j} \bm Z^{(j)} + \bm \epsilon_{j}, \ j = 2, \ldots, M,
\end{align}
where $\bm Z^{(j)} = ( (\bm X^{(1)})', (\bm X^{(2)})', \ldots, (\bm X^{(j-1)})')'$ and
$\bm A_{j} = (\bm A_{j1}, \ldots, \bm A_{j, j-1})$ with
$\bm A_{ji}$ being the $p_j \times p_i$ coefficient matrix.
Here $\bm \epsilon_j$ is the $p_j$-dimensional vector of error term for the $j$th multivariate regression with $E \bm \epsilon_j = \bm 0$ and $Cov(\bm \epsilon_j) = \bm D_j$.
Hence, we can construct a block lower triangular matrix $\bm A$ written as
\begin{align*}
\bm A = \left (
\begin{array}{ccccc}
\bm 0 & \bm 0 & \bm 0 & \ldots & \bm 0 \\
\bm A_{21} & \bm 0 & \bm 0 & \ldots & \bm 0 \\
\bm A_{31} & \bm A_{32} & \bm 0 & \ldots & \bm 0 \\
\vdots & \vdots & \ddots & \vdots & \vdots \\
\bm A_{M1} & \bm A_{M2} & \ldots & \bm A_{M, M-1} & \bm 0
\end{array}
\right),
\end{align*}
with its $j$th diagonal element being a $p_j \times p_j$ zero matrix,
and its lower left part composed of all the regression coefficient matrices in Equation \eqref{inverse covariance:eq1}.
Besides, define
\begin{align*}
\bm D_{j} = Cov(\bm \epsilon_{j}) = \left\{
\begin{array}{l}
Cov(\bm X^{(1)}), ~~~~~~~~~~~~~~~~~~~~~~~ j = 1, \\
Cov(\bm X^{(j)} - \sum_{i=1}^{j-1} \bm A_{ji} \bm X^{(i)}), ~ j = 2, 3, \ldots, M.
\end{array}
\right.
\end{align*}
Denote by $\bm D = \diag(\bm D_{1}, \bm D_{2}, \ldots, \bm D_{M})$ the block diagonal covariance matrix of vector $\bm \epsilon  = (\bm \epsilon_{1}', \bm \epsilon_{2}', \ldots, \bm \epsilon_{M}')'$.
Thus the multivariate regressions in \eqref{inverse covariance:eq1} can be written as
\begin{align}\label{inverse covariance:eq2}
\bm \epsilon = \bm X - \bm A \bm X = (\bm I- \bm A) \bm X \triangleq \bm T \bm X,
\end{align}
where $\bm I$ represents the $p \times p$ identity matrix, and $\bm T = \bm I - \bm A$ is a unit block lower triangular matrix having ones on its diagonal.
The matrices $\bm T$ and $\bm D$ are called the block Cholesky factors.
By taking $Cov(\bm \epsilon) = Cov (\bm T \bm X)$ in Equation \eqref{inverse covariance:eq2}, we have $\bm D  = \bm T \bm \Sigma \bm T'$, which consequently leads to
\begin{align}\label{decomposition}
\bm \Omega =  \bm \Sigma^{-1} = \bm T' \bm D^{-1} \bm T,
\end{align}
where $\bm D^{-1} = \diag(\bm D_{1}^{-1}, \bm D_{2}^{-1}, \ldots, \bm D_{M}^{-1})$.
As a result, the BCD method reduces the challenge of modeling an inverse covariance matrix based on the partial information of the variable ordination into
the problem of estimating $(M - 1)$ multivariate linear regressions.
Because of the decomposition \eqref{decomposition},
the general sparsity in $\bm T$ would induce some sparsity in $\bm \Omega$ (although it may not induce certain structured sparsity), which is easily implemented by the regularization on the linear regressions in \eqref{inverse covariance:eq1}.
The decomposition \eqref{decomposition} also indicates that the BCD method can guarantee the positive definiteness of the resulting estimate of $\bm \Omega$ provided that $\bm D_1^{-1}, \bm D_2^{-1}, \ldots, \bm D_{M}^{-1}$ are all positive definite.
Note that there is no constraint required for parameters in matrix $\bm T$.
We would like to remark that although the proposed model needs information on the group ordering, it is invariant to the permutation of variables within each group because the coefficient estimation of multivariate regressions in \eqref{inverse covariance:eq1} is not affected by the ordering of variables in $\bm X^{(j)}$ and $\bm Z^{(j)}$.
This point is verified in the simulation study.
Moreover, it is seen that the MCD is a special case of the proposed BCD with $M = p$ (i.e., $p_j = 1$ for $j = 1, 2, \ldots, M$),
which represents the case where the $p$ variables have a full ordering information.

\subsection{Parameter Estimation}

Denote by $\bm x_{1}, \bm x_{2}, \ldots, \bm x_{n}$ the $n$ independently and identically distributed observations from the multivariate normal distribution $\mathcal{N}_{p}(\bm 0, \bm \Omega^{-1})$.
Let $\mathbb{X} = (\bm x_{1}, \bm x_{2}, \ldots, \bm x_{n})'$ be the $n \times p$ data matrix.
Based on the partial information of the variable ordination,
we partition $\mathbb{X}$ by columns and write $\mathbb{X} = (\mathbb{X}^{(1)}, \mathbb{X}^{(2)}, \ldots, \mathbb{X}^{(M)})$,
where $\mathbb{X}^{(j)}$ represents the $n \times p_{j}$ sub-data matrix corresponding to the $j$th variable group $\bm X^{(j)}$.
Based on the methodology of BCD, we need to model a set of multivariate regressions in \eqref{inverse covariance:eq1} to obtain the estimates of block Cholesky factors matrices $\bm T$ and $\bm D^{-1}$.
The negative joint log-likelihood function is expressed as
\begin{align}\label{inverse covariance:eq5}
L(\bm A,\bm D^{-1}) = \sum_{j=1}^{M} \left \{ -\log |\bm D_{j}^{-1}| +  \tr \left[ \bm S_{\epsilon_j} \bm D_{j}^{-1} \right ] \right \},
\end{align}
up to some constant.
The symbol $\bm S_{\epsilon_j} = \frac{1}{n}(\mathbb{X}^{(j)} - \mathbb{Z}^{(j)} \bm A'_{j})'(\mathbb{X}^{(j)} - \mathbb{Z}^{(j)} \bm A'_{j})$, where  $\mathbb{Z}^{(j)} = (\mathbb{X}^{(1)}, \mathbb{X}^{(2)}, \ldots, \mathbb{X}^{(j-1)})$ stands for the data matrix of the first $(j - 1)$ groups of variables.
By encouraging the sparsity in the estimates $\hat{\bm A}_j$ and $\hat{\bm D}_j^{-1}$, we can obtain a sparse estimate $\hat{\bm \Omega} = \hat{\bm T}' \hat{\bm D}^{-1} \hat{\bm T}$.
Thus, the penalized log-likelihood function in the following is considered for parameter estimation as
\begin{align}\label{eq:obj}
L_{\lambda}(\bm A,\bm D^{-1})
&= \sum_{j=1}^{M} \left \{ -\log |\bm D_{j}^{-1}| +  \tr \left[ \bm S_{\epsilon_j} \bm D_{j}^{-1} \right ] \right \} + \lambda_{1} \sum_{j=2}^{M} \| \bm A_{j} \|_{1} +  \lambda_{2} \sum_{j=1}^{M} \| \bm D_{j}^{-1} \|_{1}^{-}   \nonumber \\
&= \sum_{j=1}^{M} \left\{ -\log |\bm D_{j}^{-1}| + \tr \left[ \bm S_{\epsilon_j} \bm D_{j}^{-1} \right ] + \lambda_{1} \| \bm A_{j} \|_{1} +  \lambda_{2} \| \bm D_{j}^{-1} \|_{1}^{-} \right\}    \nonumber \\
& \triangleq \sum_{j=1}^{M} \ell_{\lambda}(\bm A_j,\bm D_j^{-1}),
\end{align}
where $\ell_{\lambda}(\bm A_j,\bm D_j^{-1}) = -\log |\bm D_{j}^{-1}| + \tr \left[ \bm S_{\epsilon_j} \bm D_{j}^{-1} \right ] + \lambda_{1} \| \bm A_{j} \|_{1} +  \lambda_{2} \| \bm D_{j}^{-1} \|_{1}^{-}$ with $\bm A_{1} = \bm 0$ being zero matrix.
Here $\lambda_{1} \ge 0$ and $\lambda_{2} \ge 0$ are tuning parameters,
the matrix norm $||\bm B||_1 = \sum_{i,j} |b_{ij}|$, and $||\bm B||_1^{-} = \sum_{i \neq j} |b_{ij}|$ with $b_{ij}$ being the elements of matrix $\bm B$.

Note that different components $\ell_{\lambda}(\bm A_j,\bm D_j^{-1})$ contain different parameters $\bm A_j$ and $\bm D_j^{-1}$.
Hence minimizing $L_{\lambda}(\bm A,\bm D^{-1})$ is equivalent to minimizing each  $\ell_{\lambda}(\bm A_j,\bm D_j^{-1})$ separately,
which facilitates a parallel computing procedure for estimating $(\bm A_j,\bm D_j^{-1})$ simultaneously to save computational time.
For each minimization of $\ell_{\lambda}(\bm A_j,\bm D_j^{-1})$, although the objective function is not convex with respect to $(\bm A_j,\bm D_j^{-1})$, it is a biconvex optimization \citep{gorski2007biconvex}.
That is, $\ell_{\lambda}(\bm A_j,\bm D_j^{-1})$ is convex over $\bm A_j$ when fixing $\bm D_j^{-1}$, and is also convex over $\bm D_j^{-1}$ when fixing $\bm A_j$.
This property enables us to apply a coordinate descent algorithm to iteratively estimate $\bm A_j$ by minimizing $\ell_{\lambda}(\bm A_j,\bm D_j^{-1})$ for a given $\bm D_j^{-1} = (\bm D_j^{-1})_\ast$,
and estimate $\bm D_j^{-1}$ by minimizing $\ell_{\lambda}(\bm A_j,\bm D_j^{-1})$ for a given $\bm A_j = (\bm A_j)_\ast$ \citep{Rothman2010sparse, Sofer2014Variable}.
Specifically, for a given $(\bm D_j^{-1})_\ast$, we solve
\begin{align}\label{inverse covariance:eq3}
\hat{\bm A}_{j} [(\bm D_{j}^{-1})_\ast] &= \arg \min_{\bm A_j} \ell_{\lambda}(\bm A_j | (\bm D_{j}^{-1})_\ast) \nonumber  \\
&= \arg \min_{\bm A_j} \{ \frac{1}{n} \tr [(\mathbb{X}^{(j)} - \mathbb{Z}^{(j)} \bm A'_{j}) (\bm D_{j}^{-1})_\ast (\mathbb{X}^{(j)} - \mathbb{Z}^{(j)} \bm A'_{j})'] + \lambda_{1} \| \bm A_{j} \|_{1}  \}  \nonumber  \\
&= \arg \min_{\tilde{\bm A}_j} \{ \frac{1}{n}(\tilde{\mathbb{X}}^{(j)} - \tilde{\mathbb{Z}}^{(j)} \tilde{\bm A}_{j})'(\tilde{\mathbb{X}}^{(j)} - \tilde{\mathbb{Z}}^{(j)} \tilde{\bm A}_{j}) + \lambda_{1} \| \tilde{\bm A}_{j} \|_{1} \},
\end{align}
where $\tilde{\mathbb{X}}^{(j)} = \vec[\mathbb{X}^{(j)} (\bm D_{j}^{-\frac{1}{2}})_\ast]$ , $\tilde{\mathbb{Z}}^{(j)} = (\bm D_{j}^{-\frac{1}{2}})_\ast \otimes \mathbb{Z}^{(j)}$ and $\tilde{\bm A}_{j} = \vec[\bm A'_{j}]$.
Here vec denotes the vectorization operator and $\otimes$ the Kronecker product, and their property $\vec(\bm A \bm B \bm C) = (\bm C' \otimes \bm A) \vec(\bm B)$ is applied.
The optimization problem \eqref{inverse covariance:eq3} is thus a linear regression with Lasso penalty \citep{tibshirani1996regression}.
The initial value for $(\bm D_j^{-1})_\ast$ is set to be the identity matrix.
On the other hand, for a given $(\bm A_j)_\ast$, we solve
\begin{align}\label{inverse covariance:eq4}
\hat{\bm D}_j^{-1} [(\bm A_{j})_\ast] &= \arg \min_{\bm D_j^{-1}} \ell_{\lambda}(\bm D_j^{-1} | (\bm A_{j})_\ast) \nonumber  \\
&= \arg \min_{\bm D_j^{-1}} \left\{ -\log |\bm D_{j}^{-1}| + \tr \left[ (\bm S_{\epsilon_j})_\ast \bm D_{j}^{-1} \right ] + \lambda_{2} \| \bm D_{j}^{-1} \|_{1}^{-} \right\},
\end{align}
where $(\bm S_{\epsilon_j})_\ast = \frac{1}{n} [\mathbb{X}^{(j)} - \mathbb{Z}^{(j)} (\bm A'_{j})_\ast]'[\mathbb{X}^{(j)} - \mathbb{Z}^{(j)} (\bm A'_{j})_\ast]$.
It has the same form as Glasso estimation.
Accordingly, the estimates $\hat{\bm A}_j$ and $\hat{\bm D}_j^{-1}$ are iteratively solved from optimization problems \eqref{inverse covariance:eq3} and \eqref{inverse covariance:eq4} until convergence.
After obtaining the estimates $\hat{\bm A}_j$ and $\hat{\bm D}_j^{-1}$, we construct the block Choleksy factor estimate $\hat{\bm T}$ with $-\hat{\bm A}_j$ as the $j$th block row and the identity matrix being the block diagonal.
The estimate $\hat{\bm D}^{-1}$ is constructed with $\hat{\bm D}_j^{-1}$ as its $j$th block diagonal.
Then $\hat{\bm \Omega} = \hat{\bm T}' \hat{\bm D}^{-1} \hat{\bm T}$ is a sparse estimate of inverse covariance matrix under the partial information of variable ordination.
We briefly summarize the above estimation procedure in Algorithm 1.

\noindent \textbf{Algorithm 1}

\noindent Input: Data $\mathbb{X}$, tuning parameters $\lambda_{1}$ and $\lambda_{2}$.

\noindent Output: Estimate $\hat{\bm \Omega}(\lambda_{1}, \lambda_{2})$ corresponding to $\lambda_{1}$ and $\lambda_{2}$.

\noindent For $j = 1$ to $M$ do

{Step 0}: Set an initial value of $\bm D_j = \bm I$.

{Step 1}: Given $\bm D_j^{-1} = \hat{\bm D}_{j;t}^{-1}$, solve $\bm A_j$ in \eqref{inverse covariance:eq3} by the Lasso technique.

{Step 2}: Given $\bm A_j= \hat{\bm A}_{j;t}$, solve $\bm D_j^{-1}$ in \eqref{inverse covariance:eq4} by the Glasso technique.

{Step 3}: Repeat Steps 1 and 2 till both $\hat{\bm A}_j = \hat{\bm A}_{j;t}$ and $\hat{\bm D}_j^{-1} = \hat{\bm D}_{j;t}^{-1}$ converge.

\noindent End

\noindent {Step 4}: $\hat{\bm T} = \bm I - \hat{\bm A}$ and $\hat{\bm D}^{-1} = \diag(\hat{\bm D}_{1}^{-1}, \ldots, \hat{\bm D}_{M}^{-1})$, then $\hat{\bm \Omega}(\lambda_{1}, \lambda_{2}) = \hat{\bm T}' \hat{\bm D}^{-1} \hat{\bm T}$.

Here $\hat{\bm A}_{j;t}$ and $\hat{\bm D}_{j;t}$ represent the estimates of $\bm A_j$ and $\bm D_j$ in the $t$th iteration.
The convergence criteria are $||\hat{\bm A}_{j;t} - \hat{\bm A}_{j;t-1}||_F^2<\tau_1$ and $||\hat{\bm D}_{j;t} - \hat{\bm D}_{j;t-1}||_F^2<\tau_2$,
where $\tau_1$ and $\tau_2$ are two pre-selected small quantities, and $||\cdot||_F$ stands for the Frobenius norm.
Since the objective \eqref{eq:obj} is not joint convex, there is no guarantee of finding the global minimum.
However, Algorithm 1 uses a coordinate descent to compute a local solution of \eqref{eq:obj}.
Steps 1 and 2 both ensure a decrease in the value of objective, leading to the convergence of $\bm A_j$ and $\bm D_j$.
In Step 3, we also set a maximum number of iterations as 100 in case that the parameter estimation is not empirically converged.
However, in the simulation we have tried, almost all the parameter estimates satisfy the convergence criteria quickly with several tens of iterations.

Note that there are two tuning parameters $\lambda_{1}$ and $\lambda_{2}$ in the objective function \eqref{eq:obj}.
To choose their optimal values, studies in the literature often suggest cross-validation, information criteria, independent validation set mechanism and so forth.
In this work the BIC (Bayesian information criterion) proposed by Yuan and Lin (2007) is adopted to determine the optimal values of tuning parameters as follows
\begin{align*}
\mbox{BIC}(\lambda_{1}, \lambda_{2}) = -\log|\hat{\bm \Omega}(\lambda_{1}, \lambda_{2})| + \tr[\hat{\bm \Omega}(\lambda_{1}, \lambda_{2}) \bm S]
+ \frac{\log n}{n} \upsilon(\hat{\bm \Omega}(\lambda_{1}, \lambda_{2})),
\end{align*}
where $\bm S$ is the sample covariance matrix, and $\upsilon(\hat{\bm \Omega}(\lambda_{1}, \lambda_{2}))$ represents the number of non-zeros in the lower triangular part of estimate $\hat{\bm \Omega}(\lambda_{1}, \lambda_{2})$.
The optimal values of tuning parameters are selected to minimize BIC$(\lambda_{1}, \lambda_{2})$.

At the end of this section, we would like to point out that the literatures studying the MCD often assume a certain type of sparse structure, e.g. banded structure, for the underlying inverse covariance matrix.
Then some sparse patterns in $\hat{\bm T}$ will lead to certain desired sparse structures of estimate $\hat{\bm \Omega}$.
However when the underlying matrix have a general unstructured sparsity, the relationship of sparsity between matrices $\hat{\bm T}$ and $\hat{\bm \Omega}$ is not very explicit.
Nonetheless, the sparsity in $\hat{\bm \Omega}$ is also able to be induced by a sparse estimate $\hat{\bm T}$ empirically \citep{huang2006covariance, Kang2020An}.
Moreover, we justify this point in the simulation by FSL criterion which evaluates the capability of catching the sparsity of underlying matrix.

\subsection{Comparison and Connection with Several Existing Methods}
The proposed BCD method obtains a sparse inverse covariance matrix estimate under the partial information of variable ordination.
Both of the Cholesky factors' estimates $\hat{\bm T}$ and $\hat{\bm D}^{-1}$ contain the sparsity resulting from a set of penalized multivariate regressions.
It makes close connection with several existing methods.


Firstly, we demonstrate that the MCD, Glasso, \cite{Rothman2010A} and \cite{Witten2011new} are all special cases of the proposed method.
The MCD method for estimating $\bm \Omega$ is a special case when $M = p$, indicating that there is a full ordering information.
When data have partial ordering information, the MCD is not suitable.
The MCD needs a pre-specified full ordering before analyzing data, but the variables have no ordering within each group.
Therefore, one needs to identify an ordering before applying the MCD.
However, different orderings would lead to different estimates \citep{Chang2010Estimation}, and an incorrectly identified ordering would result in an inaccurate estimate.
Additionally, the Glasso is a special case of the proposed BCD with $M = 1$, implying that there is no ordering information.
Besides, \cite{Rothman2010A} studied a banded estimate of $\bm \Omega$ via MCD.
Their approach can also be viewed as a special case of the proposed BCD model by $M=p$ and regressing
$\bm X^{(j)}$ only on its several nearest previous group variables instead of all the previous group variables in Equation \eqref{inverse covariance:eq1}.
Furthermore, the proposed BCD method can be easily extended to estimate the inverse covariance matrix with a banded block structure.

In addition, \cite{Witten2011new} introduced a block diagonal inverse covariance estimation with
each block obtained by the Glasso on the corresponding group variables.
Their method assumed that variables in different groups are independent.
Thus their estimator is also a special case of the proposed BCD, in the sense that the block Cholesky factor $\bm T$ becomes the identity matrix under the independence assumption between group variables,
and each $\bm D_j^{-1}$ is then estimated from Glasso on $\mathbb{X}^{(j)}$.
The estimate $\hat{\bm \Omega} = \hat{\bm T}' \hat{\bm D}^{-1} \hat{\bm T}$ is the same as that in \cite{Witten2011new}.

Secondly, we compare several methods from perspective of ordering information.
When there is no ordering among variables, apart from Glasso which penalizes likelihood function, some papers investigated matrix estimation through penalized pseudo-likelihood.
They solve their optimization problems often in a column-by-column fashion, such as \cite{yuan2010high, cai2011constrained, liu2015fast, liu2017tiger} and so forth.
\cite{liu2015fast} extended the idea of \cite{cai2011constrained} and proposed SCIO estimator as
\begin{align*}
\hat{\bm \beta}_i = \arg \min_{\bm \beta} \left\{ \frac{1}{2} \bm \beta' \bm S \bm \beta - \bm e'_i \bm \beta + \lambda_i ||\bm \beta||_1 \right\},
\end{align*}
where $\bm e_i$ is the $i$th column of the identity matrix, $\lambda_i > 0$ is a tuning parameter, and $\hat{\bm \beta}_i$ is the estimate of the $i$th column of $\bm \Omega$.
However the estimates obtained from a column by column fashion are not guaranteed to be positive definite, and a symmetrization step is also needed to make their estimates symmetric.
In the contrast, the proposed BCD estimate is itself symmetric and positive definite.

On the other hand, when variables have a full information on the ordering, two recent works \cite{yu2017learning} and \cite{khare2019scalable} proposed to estimate $\bm \Omega$ via the classical Cholesky decomposition $\bm \Omega = \bm L' \bm L$, where the Cholesky factor $\bm L$ is a lower triangular matrix.
Both of such decomposition and the MCD induce the sparse estimates by the sparsity in the Cholesky factors.
The advantage of classical Cholesky decomposition is to directly result in a convex objective with respect to parameter $\bm L$, hence guaranteeing a global convergence with an appropriate penalty.
But it lacks statistical meanings and interpretation as explicit as the MCD technique.
\cite{yu2017learning} assumed ``local dependence" in the ordered data in the sense that the $j$th variable is correlated with its $K_j$ nearest variables, where $K_j$ can be different. That is, they assumed structured sparse pattern in the underlying $\bm \Omega$, while the proposed BCD is suitable for a general or unstructured sparsity.
\cite{khare2019scalable} introduced CSCS estimator, which accommodates the unstructured sparsity by imposing a Lasso-type penalty on the likelihood function in terms of $\bm L$, and developed a cyclic coordinate algorithm which leads to a closed form of solution for the estimates of each row of $\bm L$.
To the best of our knowledge, few works have contributed to the inverse covariance estimation when variables have partial ordering information.

\section{Theoretical Properties}\label{sec:theory}
In this section, we establish the asymptotically theoretical properties for the proposed BCD estimator.
To facilitate the presentation and proofs, we introduce some notation and make assumptions on the true model.
Let $\bm \Omega_{0} = \bm T_{0}' \bm D_{0}^{-1} \bm T_{0}$ be the underlying inverse covariance matrix with its block MCD according to the group variables $\bm X^{(1)}, \bm X^{(2)}, \ldots, \bm X^{(M)}$.
Let $\bm T_{j} = - \bm A_{j}$,
and $\bm T_{j0}$ be the counterpart of $\bm T_{j}$ in the block Cholesky factor $\bm T_{0}$.
That is, $\bm T_{j0}$ is the $j$th block row in the lower triangular part of matrix $\bm T_{0}$.
Define $\mathcal{Z}_{T_j} = \{(i, k): (\bm T_{j0})_{ik} \neq 0 \}$ as the collection of nonzero elements in the matrix $\bm T_{j0}$.
Similarly, denote the counterpart of $\bm D_j$ in the block Cholesky factor $\bm D$ by the matrix $\bm D_{j0}$, which is the $j$th block diagonal of matrix $\bm D_{0}$.
Let $\mathcal{Z}_{D_j} = \{(i, k): i \neq k, (\bm D_{j0}^{-1})_{ik} \neq 0 \}$ be the collection of nonzero off-diagonal elements in the matrix $\bm D_0^{-1}$.
Denote by $s_{T_j}$ and $s_{D_j}$ the cardinality of $\mathcal{Z}_{T_j}$ and $\mathcal{Z}_{D_j}$, respectively.
Let $s_{T} = \sum_{j=1}^M s_{T_j}$ and $s_{D} = \sum_{j=1}^M s_{D_j}$.
In order to achieve the asymptotic consistent property of the proposed estimator,
a mild condition is needed that there exists a constant $\theta > 0$ such that the singular values of $\bm \Omega_{0}$ are bounded as
\begin{align}\label{assumption1}
1/\theta < \varphi_{p}(\bm \Omega_{0}) \leq \varphi_{1}(\bm \Omega_{0}) < \theta,
\end{align}
where $\varphi_{1}(\bm B), \varphi_{2}(\bm B), \ldots, \varphi_{p}(\bm B)$ represent the singular values of matrix $\bm B$ in a decreasing order.
This assumption is made to guarantee the positive definiteness of $\bm \Omega_0$.
In addition, assume that there exist constants $0 \leq C_j \leq 1$ and $\gamma_j \geq 1$ such that
\begin{align}\label{assumption2}
p_j = \gamma_j p^{C_j} ~~~\mbox{for}~~~j = 1, 2, \ldots, M.
\end{align}
Assumption \eqref{assumption2} characterizes the relation between each $p_j$ and $p$ for high-dimensional data.
If $p$ diverges to infinity, $C_j = 0$ corresponds to the case where $p_j$ is fixed and controlled by $\gamma_j$, while $0 < \tau_c \leq C_j \leq 1$ implies that $p_j$ also diverges, but not faster than $p$, where $\tau_c$ is some positive arbitrarily small number.
Denote $\mathcal{Z}_{C} = \{j: 0 < \tau_c \leq C_j \leq 1 \}$, the set of indices that the corresponding groups have a diverging number of variables.
Now, we present the main results in Theorems \ref{theorem1} and \ref{theorem2}.

\begin{theorem}{\label{theorem1}}
Suppose that $\bm x_{1}, \ldots, \bm x_{n}$ are $n$ independently and identically distributed observations from $\mathcal{N}_{p}(\bm 0, \bm \Omega^{-1})$.
Let $(\bm D_{j}^{-1})_\ast$ and $(\bm A_{j})_\ast$ be any estimates of $\bm D_j^{-1}$ and $\bm A_j$ obtained from the path of Algorithm 1.
Under \eqref{assumption1} and \eqref{assumption2}, assume that the tuning parameters $\lambda_{1}$ and $\lambda_{2}$ satisfy $\lambda_{1} = O(\sqrt{\log(p) / n})$, $\lambda_{2} = O(\sqrt{\log(p) / n})$, then

\noindent (a) there exists a local minimum $\hat{\bm A}_j( = -\hat{\bm T}_j )$ of $\ell_{\lambda}(\bm A_j | (\bm D_{j}^{-1})_\ast)$ such that $\| \hat{\bm T_j} - \bm T_{j0} \|_{F} = O_p(s_{T_j} \log(\sum_{k=1}^j p_k) / n)$.

\noindent (b) there exists a local minimum $\hat{\bm D}_j$ of $\ell_{\lambda}(\bm D_j^{-1} | (\bm A_{j})_\ast)$ such that $\| \hat{\bm D_j} - \bm D_{j0} \|_{F} = O_p ( (s_{D_j} + p_j) \log p_j / n)$.
\end{theorem}

\begin{theorem}{\label{theorem2}}
Let $\hat{\bm \Omega}$ be the estimate of $\bm \Omega$ obtained by Algorithm 1.
Assume all the assumptions in Theorem \ref{theorem1} hold, then we have
\begin{align*}
\| \hat{\bm \Omega} - \bm \Omega_{0} \|_{F} = O_p \left( \sqrt{\frac{s_{T} \log p + \sum_{j=1}^M (s_{D_{j}} + p_{j}) \log p_{j} }{n}} \right).
\end{align*}
Under the condition $s_{T} \log p + \sum_{j=1}^M (s_{D_{j}} + p_{j}) \log p_{j} = o(n)$,
$\| \hat{\bm \Omega} - \bm \Omega_{0} \|_{F} \stackrel{P}{\rightarrow} 0$.
\end{theorem}

Theorem \ref{theorem1} provides asymptotic consistent rates of the estimators $\hat{\bm T}_j$ and $\hat{\bm D}_j$ that are obtained by minimizing $\ell_{\lambda}(\bm A_j | (\bm D_{j}^{-1})_\ast)$ in \eqref{inverse covariance:eq3} and minimizing $\ell_{\lambda}(\bm D_j^{-1} | (\bm A_{j})_\ast)$ in \eqref{inverse covariance:eq4}.
Theorem \ref{theorem2} establishes the property for the inverse covariance estimate $\hat{\bm \Omega}$ obtained from the proposed Algorithm 1, which demonstrates the consistent rate of our estimate in practice.
Moreover, we would like to have the following remarks.

First, from Theorem \ref{theorem1} we have $\| \hat{\bm D}^{-1} - \bm D_{0}^{-1} \|_{F}^2 = \| \hat{\bm D} - \bm D_{0} \|_{F}^2 = \sum_{j=1}^M \| \hat{\bm D_j} - \bm D_{j0} \|_{F} = O_p(\sum_{j=1}^M (s_{D_{j}} + p_{j}) \log (p_{j}) / n)$.
Such rate is sharper than the theoretical Glasso consistent rate $(\tilde{s} + p) \log (p) / n$ ($\tilde{s}$ is a measure of sparsity), which has been obtained in the literature \citep{lam2009sparsistency}, because of  $\sum_{j=1}^M p_j \log p_j \leq p \log p$ resulting from $\prod_{j=1}^M p_j^{p_j} \leq \prod_{j=1}^M p_{\max}^{p_j} = p_{\max}^{p} \leq p^{p}$, where $p_{\max} = \max \{p_j\}^{M}_{j=1}$.
The equality holds only in the situation that all the $p$ variables are in one group, where it is exactly a Glasso problem.
This implies that, as long as there are at least two groups of variables, the proposed BCD model with partial information of variable ordination is useful in reducing the consistent rate.

Second, the estimation of $\bm D$ is decomposed into $M$ separate Glasso estimations in the proposed method.
According to the log sum inequality, we have $\sum_{j=1}^M p_j \log p_j \geq p \log (p/M)$.
The equality holds when $p_j = p/M$ for all $j = 1, 2, \ldots, M$, indicating that the lowest bound of the proposed estimator's consistent rate would be achieved when each group has the equal number of variables.
In view of this, we would like to point out that:
(i) a larger number of variable groups leads to a smaller value of the lowest bound of the consistent rate;
(ii) the more evenly that $p$ variables are assigned into $M$ groups, the more closely that the consistent rate tends to the lowest bound.

Third, the requirement $s_{T} \log p + \sum_{j=1}^M (s_{D_{j}} + p_{j}) \log p_{j} = o(n)$ in Theorem \ref{theorem2} is a relatively weaker condition compared with the assumption for the Glasso model.
Note that $\| \hat{\bm D} - \bm D_{0} \|^2_{F} \leq O_p(\sum_{j=1}^M (s_{D_{j}} + p_{j}) \log (p_{\max}) / n) = O_p((s_{D} + p) \log (p_{\max}) / n)$, then we may require a stronger condition $s_{T} \log p + (s_{D} + p) \log p_{\max} = o(n)$.
Moreover, if we further loose the upper bound of $\| \hat{\bm D} - \bm D_{0} \|^2_{F} \leq O_p((s_{D} + p) \log (p) / n)$, the proposed BCD model needs an even stronger condition $(s_{T} + s_{D} + p) \log p = o(n)$, which however is the similar condition as that for the Glasso estimator.

\section{Numerical Study}\label{sec:simulation}
In this section, we conduct simulation studies to evaluate the performance of the proposed BCD model (Prop) in comparison with several existing methods, including the MCD method, Glasso, SCIO \citep{liu2015fast} and CSCS \citep{khare2019scalable}.
The MCD, where $M = p$, reduces Equation \eqref{inverse covariance:eq1} to $(p-1)$ univariate linear regressions with their coefficients estimated by Lasso.
The Glasso and SCIO are two popular methods, but do not consider ordering information.
The CSCS is suitable for data with full ordering information and implemented via  classical Cholesky decomposition.
The tuning parameters in the Glasso, SCIO and CSCS methods are selected based on BIC.
Besides, to examine whether the proposed model is permutation invariant within each group, we also implement the Prop$^\ast$ method, which estimates $\bm \Omega$ by Algorithm 1 from data which randomly permutate variables within each group.

The data are independently generated from $\mathcal{N}_{p}(\bm 0, \bm \Omega^{-1})$ with sample size $n = 50$ and number of variables $p = 200$.
We consider two different cases of variable groups: (1) five groups with each containing 40 variables; (2) four groups with each subsequently containing 30, 60, 40 and 70 variables.
Let $\mbox{AR}(\rho)$ represent a squared matrix with autoregressive structure with $(i,j)$th entry as $\rho^{|i - j|}$, $1 \leq i, j \leq p$.
Let $\mbox{MA}(0.5,0.4,0.3)$ indicate a squared banded matrix with the main diagonal elements 1, and the subsequent sub-diagonal elements are 0.5, 0.4 and 0.3 respectively.
Denote an $a \times b$ non-squared matrix $\widetilde{\mbox{AR}}(\rho) = \left( \mbox{AR}(\rho), \bm 0 \right)$ if $a < b$, and $\widetilde{\mbox{AR}}(\rho) = \left(
\begin{array}{ccccc}
\mbox{AR}(\rho) \\
\bm 0
\end{array}
\right)$ otherwise, where $\bm 0$ represents the matrix with all elements 0.
Similarly, Denote an $a \times b$ non-squared matrix $\widetilde{\mbox{MA}}(0.5,0.4,0.3) = \left( \mbox{MA}(0.5,0.4,0.3), \bm 0 \right)$ if $a < b$, and $\widetilde{\mbox{MA}}(0.5,0.4,0.3) = \left(
\begin{array}{ccccc}
\mbox{MA}(0.5,0.4,0.3) \\
\bm 0
\end{array}
\right)$ otherwise.
To systematically investigate the performance of the proposed method, we consider the following different structures of inverse covariance matrix $\bm \Omega$.

\begin{itemize}
\item $\textbf{Scenario 1}$. $\bm \Omega_{1} = \mbox{AR}(0.8)$.


\item $\textbf{Scenario 2}$. $\bm \Omega_{2} = \left(
\begin{array}{cccccccc}
    \mbox{AR}(0.5)            & \ldots & \bm 0 \\
    \vdots                    & \ddots & \bm 0 \\
    \bm 0                     &\ldots  & \mbox{AR}(0.5)
  \end{array}
\right)$.

\item $\textbf{Scenario 3}$.

$\bm \Omega_{3} = \left(
\begin{array}{cccccccc}
\mbox{MA}(0.5,0.4,0.3)&\widetilde{\mbox{AR}}(0.5)&\ldots&\widetilde{\mbox{AR}}(0.5) \\
\widetilde{\mbox{AR}}(0.5)&\mbox{MA}(0.5,0.4,0.3)&\ldots&\widetilde{\mbox{AR}}(0.5) \\
\vdots&\widetilde{\mbox{AR}}(0.5)&\ddots&\widetilde{\mbox{AR}}(0.5) \\
\widetilde{\mbox{AR}}(0.5)&\widetilde{\mbox{AR}}(0.5)&\vdots&\mbox{MA}(0.5,0.4,0.3)
\end{array}
\right)$ + $\bm \alpha \bm I$.
The value of $\bm \alpha$ is gradually increased to ensure that $\bm \Omega_{3}$ is positive definite.

\item $\textbf{Scenario 4}$. $\bm \Omega_{4} = \tilde{\bm \Omega}_{4} + \bm \alpha \bm I$, where $\tilde{\bm \Omega}_{4}$ is generated by randomly permuting rows and corresponding columns of each block of \\
$\left(
\begin{array}{cccccccc}
\mbox{AR}(0.5)&\widetilde{\mbox{MA}}(0.5,0.4,0.3)&\ldots&\widetilde{\mbox{MA}}(0.5,0.4,0.3) \\
\widetilde{\mbox{MA}}(0.5,0.4,0.3)&\mbox{AR}(0.5)&\ldots&\widetilde{\mbox{MA}}(0.5,0.4,0.3) \\
\vdots& \widetilde{\mbox{MA}}(0.5,0.4,0.3)&\ddots&\widetilde{\mbox{MA}}(0.5,0.4,0.3) \\
\widetilde{\mbox{MA}}(0.5,0.4,0.3)&\widetilde{\mbox{MA}}(0.5,0.4,0.3)&\vdots&\mbox{AR}(0.5)
\end{array}
\right)$.
The value of $\bm \alpha$ is gradually increased to ensure that $\bm \Omega_{4}$ is positive definite.

\item $\textbf{Scenario 5}$. $\bm \Omega_{5} = \bm B' \bm H \bm B$, where $\bm H$ is a block diagonal matrix with its each diagonal block as AR(0.5), and $\bm B = (b_{i,j})$ with $b_{i,i} = 1, b_{i+20,i} = -0.8$ and $b_{i,j} = 0$ otherwise.

\item $\textbf{Scenario 6}$. $\bm \Omega_{6} = \bm B' \bm H \bm B$, where $\bm H$ is a block diagonal matrix with its each diagonal block being MA(0.5,0.4,0.3), and $\bm B = (b_{i,j})$ with $b_{i,i} = 1, b_{i+20,i} = -0.8, b_{i+21,i} = 0.5$ and $b_{i,j} = 0$ otherwise.

\item $\textbf{Scenario 7}$. $\bm \Omega_{7} = \tilde{\bm \Omega}_{7} + \alpha \bm I$. Here the diagonal elements of $\tilde{\bm \Omega}_{7}$ are 0, and each off-diagonal element is generated independently as $(\tilde{\bm \Omega}_{7})_{ij} = (\tilde{\bm \Omega}_{7})_{ji} = b * Unif(-1, 1)$, where $b$ is a Bernoulli random variable with probability 0.15 equal 1.
    The value of $\alpha$ increases gradually to make sure $\bm \Omega_{7}$ is positive definite.
\end{itemize}

$\textbf{Scenario 1}$ is the AR structure with the variables' correlations decaying when they are far apart from each other.
$\textbf{Scenarios 2}$ is a block diagonal matrix with multiple groups of variables, where the variables in different groups are independent.
$\textbf{Scenarios 3}$ and $\textbf{4}$ are block matrices with multiple groups of variables with variables in different groups possibly correlated.
$\textbf{Scenarios 5}$ and $\textbf{6}$ are similarly used in \cite{huang2006covariance}.
$\textbf{Scenario 7}$ is a general sparse matrix with random structure.

To evaluate the accuracy of each estimate $\hat{\bm \Omega} = (\hat{\omega}_{ij})$ for the underlying inverse covariance matrix $\bm \Omega = (\omega_{ij})$, we consider the loss measures $L_1$, the matrix spectral norm $L_2$, the Frobenius norm F of $(\bm \Omega - \hat{\bm \Omega})$ as follows
\begin{align*}
L_1 = \max_{ j } \sum_{i} | \hat{\omega}_{ij} - \omega_{ij} |, ~~~
L_2 = \lambda_{\max}[(\bm \Omega - \hat{\bm \Omega})], ~~~
\mbox{Fnorm} = \sqrt{\sum_{i=1}^{p} \sum_{j=1}^{p} (\hat{\omega}_{ij} - \omega_{ij})^2},
\end{align*}
where $\lambda_{\max}$ is the maximum value of eigenvalues.
We also use the Kullback-Leibler loss (KL) as well as the quadratic loss (QL)
\begin{align*}
\mbox{KL} = \frac{1}{p}~(\tr [\bm \Omega^{-1} \hat{\bm \Omega}] - \log |\bm \Omega^{-1} \hat{\bm \Omega}| - p), ~~~
\mbox{QL} = \frac{1}{p}~\tr (\bm \Omega^{-1} \hat{\bm \Omega} - \bm I)^2.
\end{align*}
In addition, to gauge the ability of the proposed model to capture the underlying sparse structure, we report the false selection loss FSL = (FP + FN) / $p^2$ in percentage, where FP is the false positive and FN is the false negative.
The simulation results of loss measures for each method are summarized in Tables 1 and 2, reporting their averages and corresponding standard errors (in parenthesis) over 50 replicates.

\begin{table}
\begin{center}
\caption{The averages and standard errors of estimates for 5 groups of variables.}\label{table:5group}
\resizebox{\textwidth}{!}{ 
\begin{tabular}{crrrrrrrrrrrrrrrr}
\hline
\textbf{Scenario}& &$L_1$ &$L_2$ &Fnorm &KL &QL &FSL (\%)\\\hline
\multirow {6}*{1}
&MCD        & 50.25 (6.628) & 16.66 (2.187) & 196.5 (24.94) & 4.000 (0.670) & 4.738 (1.027) & 34.75 (0.302) \\
&Glasso     & 8.852 (0.001) & 8.717 (0.001) & 28.61 (0.004) & 0.518 (0.001) & 0.601 (0.004) & 27.67 (0.001) \\
&SCIO       & 8.843 (0.002) & 8.685 (0.001) & 28.43 (0.004) & 0.451 (0.003) & 0.655 (0.005) & 27.35 (0.006) \\
&CSCS       & 8.750 (0.004) & 8.577 (0.002) & 28.06 (0.009) & 0.415 (0.002) & 0.505 (0.005) & 27.30 (0.019) \\
&Prop       & 8.768 (0.004) & 8.581 (0.001) & 27.87 (0.006) & 0.363 (0.001) & 0.397 (0.004) & 25.89 (0.022) \\
\midrule
\multirow {6}*{2}
&MCD        & 41.23 (15.07) & 32.46 (4.881) & 34.83 (5.265) & 0.366 (0.054) & 2.921 (0.699) & 29.53 (0.580) \\
&Glasso     & 2.613 (0.004) & 2.379 (0.003) & 12.57 (0.009) & 0.242 (0.001) & 0.410 (0.003) & 7.875 (0.001) \\
&SCIO       & 2.618 (0.005) & 2.363 (0.003) & 12.48 (0.010) & 0.240 (0.003) & 0.419 (0.003) & 7.769 (0.003) \\
&CSCS       & 2.586 (0.006) & 2.324 (0.005) & 12.20 (0.031) & 0.222 (0.002) & 0.390 (0.003) & 7.711 (0.010) \\
&Prop       & 2.539 (0.007) & 2.259 (0.002) & 11.78 (0.008) & 0.195 (0.001) & 0.340 (0.003) & 7.343 (0.006) \\
\midrule
\multirow {6}*{3}
&MCD        & 61.20 (5.725) & 205.6 (20.88) & 245.2 (25.91) & 4.671 (0.303) & 28.53 (6.462) & 49.24 (0.163) \\
&Glasso     & 9.592 (0.001) & 9.414 (0.001) & 25.87 (0.004) & 0.778 (0.001) & 2.952 (0.023) & 36.35 (0.001) \\
&SCIO       & 9.579 (0.002) & 9.348 (0.001) & 25.62 (0.004) & 0.703 (0.003) & 3.506 (0.029) & 35.99 (0.007) \\
&CSCS       & 9.543 (0.007) & 9.153 (0.005) & 25.14 (0.016) & 0.591 (0.003) & 1.604 (0.024) & 34.72 (0.032) \\
&Prop       & 9.530 (0.005) & 9.140 (0.001) & 25.09 (0.003) & 0.424 (0.001) & 0.681 (0.009) & 36.99 (0.018) \\
\midrule
\multirow {6}*{4}
&MCD        & 50.70 (6.839) & 16.30 (2.144) & 19.04 (2.179) & 3.108 (0.210) & 54.44 (11.78)  & 47.52 (0.204) \\
&Glasso     & 4.221 (0.031) & 3.474 (0.012) & 11.33 (0.111) & 0.568 (0.041) & 10.79 (2.661) & 21.14 (0.079) \\
&SCIO       & 3.936 (0.001) & 3.546 (0.001) & 11.91 (0.012) & 0.597 (0.083) & 9.475 (1.973) & 19.70 (0.001) \\
&CSCS       & 6.880 (0.140) & 3.300 (0.006) & 10.12 (0.032) & 0.566 (0.019) & 2.601 (0.067) & 22.23 (0.109) \\
&Prop       & 4.341 (0.010) & 3.373 (0.002) & 10.50 (0.008) & 0.466 (0.001) & 2.075 (0.043) & 21.12 (0.014) \\
\midrule
\multirow {6}*{5}
&MCD        & 136.7 (13.75) & 50.75 (5.017) & 67.67 (5.552) & 1.282 (0.086) & 47.12 (6.803) & 31.27 (0.281) \\
&Glasso     & 9.354 (0.005) & 8.612 (0.005) & 30.83 (0.029) & 0.626 (0.006) & 8.278 (0.366) & 23.54 (0.013) \\
&SCIO       & 9.543 (0.002) & 8.841 (0.002) & 32.11 (0.010) & 0.637 (0.008) & 3.559 (0.034) & 23.85 (0.001) \\
&CSCS       & 9.532 (0.014) & 8.540 (0.015) & 29.38 (0.075) & 0.966 (0.014) & 1.786 (0.013) & 22.77 (0.043) \\
&Prop       & 9.311 (0.007) & 8.390 (0.005) & 29.77 (0.026) & 0.614 (0.003) & 0.665 (0.085) & 23.09 (0.055) \\
\midrule
\multirow {6}*{6}
&MCD        & 34.50 (2.309) & 13.79 (0.952) & 27.24 (0.814) & 1.493 (0.016) & 20.98 (0.858)  & 25.59 (0.234) \\
&Glasso     & 8.291 (0.004) & 5.340 (0.005) & 25.74 (0.028) & 1.052 (0.009) & 1.717 (0.121) & 12.45 (0.027) \\
&SCIO       & 10.52 (0.713) & 7.234 (0.382) & 26.68 (0.478) & 1.243 (0.083) & 1.494 (0.062) & 10.70 (0.001) \\
&CSCS       & 7.987 (0.017) & 4.921 (0.008) & 21.62 (0.044) & 1.278 (0.039) & 1.275 (0.018) & 11.33 (0.111) \\
&Prop       & 7.724 (0.024) & 5.003 (0.015) & 22.92 (0.103) & 1.062 (0.008) & 0.803 (0.258) & 10.31 (0.064) \\
\midrule
\multirow {6}*{7}
&MCD        & 78.65 (9.265) & 25.19 (2.910) & 28.24 (2.983) & 4.042 (0.291) & 11.57 (2.259)  & 48.25 (0.290) \\
&Glasso     & 5.112 (0.054) & 1.638 (0.005) & 10.88 (0.041) & 0.537 (0.010) & 3.615 (0.313) & 16.97 (0.027) \\
&SCIO       & 4.502 (0.008) & 1.663 (0.001) & 10.74 (0.017) & 0.547 (0.008) & 3.227 (0.412) & 14.86 (0.010) \\
&CSCS       & 5.776 (0.099) & 1.580 (0.005) & 9.346 (0.034) & 0.524 (0.007) & 2.444 (0.081) & 17.16 (0.092) \\
&Prop       & 4.618 (0.020) & 1.540 (0.003) & 9.701 (0.027) & 0.468 (0.002) & 2.129 (0.054) & 16.88 (0.028) \\
\hline
\end{tabular}}
\end{center}
\end{table}

\begin{table}
\begin{center}
\caption{The averages and standard errors of estimates for 4 groups of variables.}\label{table:4group}
\resizebox{\textwidth}{!}{ 
\begin{tabular}{crrrrrrrrrrrrrrrr}
\hline
\textbf{Scenario}& &$L_1$ &$L_2$ &Fnorm &KL &QL &FSL (\%)\\\hline
\multirow {6}*{1}
&MCD        & 48.46 (6.760) & 16.19 (2.258) & 178.1 (23.47) & 4.400 (0.556) & 3.809 (0.841)  & 33.46 (0.316) \\
&Glasso     & 8.850 (0.001) & 8.718 (0.001) & 28.62 (0.003) & 0.518 (0.001) & 0.597 (0.003) & 27.68 (0.001) \\
&SCIO       & 8.845 (0.002) & 8.683 (0.001) & 28.43 (0.004) & 0.489 (0.002) & 0.648 (0.004) & 27.35 (0.005) \\
&CSCS       & 8.771 (0.003) & 8.603 (0.002) & 28.05 (0.010) & 0.451 (0.003) & 0.497 (0.004) & 27.29 (0.022) \\
&Prop       & 8.785 (0.003) & 8.610 (0.003) & 28.04 (0.016) & 0.387 (0.002) & 0.421 (0.003)  & 26.43 (0.030) \\
\midrule
\multirow {6}*{2}
&MCD        & 23.83 (19.01) & 36.34 (6.053) & 38.72 (6.322) & 0.397 (0.060) & 3.993 (0.879)  & 28.31 (0.577) \\
&Glasso     & 2.612 (0.004) & 2.384 (0.002) & 12.60 (0.008) & 0.243 (0.001) & 0.418 (0.003) & 8.099 (0.001) \\
&SCIO       & 2.635 (0.019) & 2.367 (0.002) & 12.51 (0.008) & 0.242 (0.003) & 0.430 (0.004) & 7.998 (0.003) \\
&CSCS       & 2.605 (0.008) & 2.331 (0.005) & 12.26 (0.036) & 0.227 (0.002) & 0.400 (0.004) & 7.976 (0.015) \\
&Prop       & 2.573 (0.008) & 2.271 (0.004) & 11.88 (0.026) & 0.201 (0.002) & 0.354 (0.004)  & 7.687 (0.013) \\
\midrule
\multirow {6}*{3}
&MCD        & 16.27 (1.882) & 60.19 (7.668) & 81.22 (7.474) & 2.385 (0.119) & 16.14 (2.668)  & 41.68 (0.223) \\
&Glasso     & 7.953 (0.003) & 7.766 (0.007) & 22.23 (0.034) & 1.147 (0.033) & 7.825 (0.492) & 21.93 (0.030) \\
&SCIO       & 7.971 (0.055) & 7.668 (0.041) & 21.14 (0.042) & 1.337 (0.054) & 9.113 (0.230) & 21.38 (0.010) \\
&CSCS       & 7.845 (0.011) & 7.451 (0.009) & 20.01 (0.032) & 0.953 (0.048) & 1.564 (0.036) & 22.53 (0.115) \\
&Prop       & 7.817 (0.006) & 7.405 (0.002) & 20.78 (0.007) & 0.486 (0.002) & 0.789 (0.016)  & 23.73 (0.021) \\
\midrule
\multirow {6}*{4}
&MCD        & 26.29 (4.073) & 8.836 (1.365) & 11.09 (1.435) & 1.883 (0.142) & 18.57 (4.659)  & 39.62 (0.273) \\
&Glasso     & 4.284 (0.051) & 2.809 (0.008) & 9.456 (0.055) & 0.347 (0.015) & 2.629 (1.172) & 15.57 (0.037) \\
&SCIO       & 3.576 (0.001) & 2.939 (0.001) & 10.22 (0.012) & 0.405 (0.026) & 2.506 (1.088) & 14.13 (0.001) \\
&CSCS       & 5.846 (0.119) & 2.684 (0.005) & 8.746 (0.024) & 0.369 (0.010) & 1.782 (0.043) & 15.88 (0.046) \\
&Prop       & 4.491 (0.025) & 2.788 (0.003) & 9.024 (0.008) & 0.358 (0.001) & 1.487 (0.033)  & 15.82 (0.017) \\
\midrule
\multirow {6}*{5}
&MCD        & 69.74 (5.203) & 25.87 (1.396) & 56.76 (1.174) & 1.714 (0.097) & 68.35 (9.419) & 36.22 (0.221) \\
&Glasso     & 25.84 (0.003) & 20.54 (0.002) & 53.62 (0.008) & 0.824 (0.005) & 11.44 (0.379) & 35.91 (0.014) \\
&SCIO       & 25.96 (0.002) & 20.65 (0.001) & 54.02 (0.002) & 0.797 (0.008) & 4.084 (0.035) & 37.30 (0.001) \\
&CSCS       & 25.62 (0.007) & 20.35 (0.003) & 52.98 (0.010) & 1.057 (0.035) & 3.884 (0.011) & 34.60 (0.035) \\
&Prop       & 25.76 (0.004) & 20.42 (0.003) & 53.07 (0.013) & 0.734 (0.002) & 3.003 (0.200) & 33.54 (0.049) \\
\midrule
\multirow {6}*{6}
&MCD        & 38.13 (3.725) & 15.74 (1.756) & 29.40 (1.459) & 1.032 (0.022) & 22.71 (1.110)  & 25.72 (0.266) \\
&Glasso     & 8.295 (0.005) & 5.364 (0.005) & 25.81 (0.031) & 1.072 (0.010) & 1.355 (0.114) & 12.55 (0.031) \\
&SCIO       & 9.359 (0.798) & 6.426 (0.358) & 27.73 (0.641) & 1.106 (0.094) & 1.379 (0.056) & 10.99 (0.019) \\
&CSCS       & 7.985 (0.022) & 4.827 (0.009) & 21.68 (0.047) & 1.091 (0.018) & 1.308 (0.021) & 11.65 (0.128) \\
&Prop       & 7.659 (0.017) & 5.005 (0.009) & 22.86 (0.070) & 0.902 (0.005) & 0.726 (0.159)  & 11.08 (0.033) \\
\midrule
\multirow {6}*{7}
&MCD        & 40.16 (5.131) & 13.34 (1.641) & 16.00 (1.795) & 2.658 (0.200) & 37.25 (8.799)  & 47.28 (0.289) \\
&Glasso     & 5.944 (0.099) & 1.588 (0.006) & 10.53 (0.060) & 0.480 (0.015) & 4.094 (0.968) & 16.34 (0.038) \\
&SCIO       & 4.682 (0.001) & 1.619 (0.001) & 10.34 (0.014) & 0.483 (0.009) & 3.901 (0.876) & 15.13 (0.001) \\
&CSCS       & 8.205 (0.198) & 1.604 (0.038) & 9.022 (0.032) & 0.451 (0.006) & 1.756 (0.068) & 17.52 (0.070) \\
&Prop       & 4.992 (0.029) & 1.503 (0.001) & 9.063 (0.010) & 0.374 (0.001) & 1.312 (0.029)  & 17.15 (0.016) \\
\hline
\end{tabular}}
\end{center}
\end{table}

Overall, the Prop method gives relatively superior performance compared with other methods for the considered loss measures.
Especially it performs substantially well with respect to KL and QL for all the scenarios.
Note that the results of Prop$^\ast$ are omitted in the tables since they are exactly the same as the results of the Prop method.
It further confirms that the proposed BCD model is permutation invariant for the variables within each group.
In scenario 1, the CSCS and Prop are the best models.
The CSCS method is slight better than Prop regarding $L_1, L_2$ and Fnorm since it is designed for inverse covariance estimation of data with full information on variable ordering.
For other scenarios 2-7, we observe that (1) the Prop is at least comparable with or slightly better than CSCS method; (2) the Prop outperforms SCIO, Glasso and MCD methods for most loss measures.
In addition, the compared methods also show their advantages in some loss functions.
For example, the SCIO produces the lowest values in terms of $L_1$ and FSL for both of scenarios 4 and 7.
The CSCS gives the best performance in terms of Fnorm for scenarios 5, 6 and 7.
It is also superior over other methods regarding $L_2$ for scenarios 4 and 6.

Besides, one can see that the MCD approach does not perform well compared with other methods, possibly due to that most simulated data do not have a valid full information of natural variable ordering. It can make the MCD, the performance of which heavily depends on the variable ordering, being inferior to other methods.
Moreover, for Scenario 7 without any variable ordering, we observe that the Glasso can be better in certain criteria such as $L_1$ and $FSL$.
Especially, the Glasso is superior in FSL as expected since it is good at inducing sparsity.
Comparing with CSCS, the Glasso gives similar performance on $L_2$, although not as good in terms of Fnorm and QL.


\section{Case Studies}\label{sec:app}
\subsection{Covid-19 Data} \label{sec:covid}
We further evaluate the performance of the proposed model through a real data example of Covid-19 pandemic, which is available on the official website of Virginia health department.
The data were weekly collected during May 29th, 2020 to February 6th, 2021 (37 weeks) from 34 districts in Virginia State such as Arlington, Fairfax, Richmond, Roanoke, Virginia Beach, etc.
For each week the data contain four variables: the accumulative number of cases, the accumulative number of people hospitalized, the accumulative number of deaths and the accumulative number of people taking the PCR (polymerase chain reaction) tests,
resulting in $4 \times 37 = 148$ variables.


Denote the collected data by $\bm N_i = (N_{i1},\ldots,N_{i148})', i = 1, . . . , 34$,
and transform $y_{ij} = \sqrt{(N_{ij} + 1/4)}$ to make the data distribution close to normal \citep{brown2005statistical}.
We then apply the proposed model as well as the MCD, SCIO, CSCS and Glasso methods to estimate the $148 \times 148$ inverse covariance matrix.
To conduct data analysis, the 148 variables are partitioned into 37 groups with each group naturally corresponding to a calendar week.
Hence it is seen that the data have a partial information of the variable ordination in a weekly scale.
To examine the performance of methods in comparison, we predict the accumulative number of cases, the accumulative number of people hospitalized, the accumulative number of deaths and the accumulative number of people taking PCR in the last 2 weeks
using the inverse covariance estimates obtained from the data in the first 35 weeks.
Specifically, let $\bm y_i = (y_{i1},\ldots,y_{i148})' = (\bm y_{iE}', \bm y_{iL}')'$, where $\bm y_{iE}$ and $\bm y_{iL}$ represent the data in the first 35 weeks and last 2 weeks for the $i$th district, $i = 1, 2, \ldots , 34$.
Hence $\bm y_{iE}$ contains $4 \times 35 = 140$ variables and $\bm y_{iL}$ contains $4 \times 2 = 8$ variables.
Accordingly, the mean vector and the inverse covariance matrix are divided as
\begin{align*}
\bm \mu = \left[
\begin{array}{cc}
\bm \mu_{1} \\
\bm \mu_{2}
\end{array}\right]~~~\mbox{and}~~~
\bm \Omega= \left[
\begin{array}{cc}
\bm \Omega_{11},      & \bm \Omega_{12} \\
\bm \Omega_{12}',   & \bm \Omega_{22}
\end{array}\right].
\end{align*}
Assuming multivariate normality, we have
\begin{align}\label{eq:pred}
E(\bm y_{iL}|\bm y_{iE}) = \bm \mu_2 - \bm \Omega_{22}^{-1} \bm \Omega_{12}' (\bm y_{iE} - \bm \mu_1).
\end{align}
The 34 observations are split into a training set and a testing set using the leaving-one-out mechanism.
That is, each observation is considered as a testing point with the rest 33 observations as training set.
The training set is used to estimate the mean vector $\bm \mu$ and the inverse covariance matrix $\bm \Omega$.
The values of $\bm y_{iE} = (y_{i1}, \ldots, y_{i140})'$ in the testing data are used to predict $\bm y_{iL} = (y_{i141}, \ldots, y_{i148})'$ based on Equation \eqref{eq:pred}.
For each variable in $\bm y_{iL}$, define the average absolute prediction error (APE) by
\begin{align*}
\mbox{APE}_j = \frac{1}{34} \sum_{i=1}^{34}|\hat{y}_{ij} - y_{ij}|,~~~j = 141, \ldots, 148.
\end{align*}
where $\hat{y}_{ij}$ is the predicted value.

\begin{table}
\begin{center}
\caption{The averages ($\times 10^{-1}$) and standard errors ($\times 10^{-1}$) of APE for Covid-19 data.} \label{table:APE}
\resizebox{\textwidth}{!}{ 
\begin{tabular}{rrrrrrrrrrrrrrrrr}
\hline
  \multirow{2}*&\multicolumn{4}{c}{36th week} &&\multicolumn{4}{c}{37th week} \\
\cline{2-10}
 &Cases     &Hospitalizations    &Deaths      &PCR &&Cases     &Hospitalizations    &Deaths      &PCR  \\
\hline
MCD     &2.90 (0.26) &0.79 (0.11) &0.79 (0.11) &1.41 (0.16) &&1.65 (0.22) &0.87 (0.12) &0.75 (0.09) &2.20 (0.40) \\
Glasso  &3.57 (0.43) &1.03 (0.13) &1.02 (0.13) &5.50 (0.68) &&2.41 (0.32) &1.00 (0.10) &0.93 (0.11) &5.16 (0.75) \\
SCIO    &2.48 (0.32) &0.95 (0.13) &1.04 (0.12) &2.65 (0.34) &&2.22 (0.31) &0.99 (0.10) &0.93 (0.11) &3.02 (0.51) \\
CSCS    &2.54 (0.24) &1.06 (0.13) &0.79 (0.12) &1.50 (0.19) &&1.64 (0.19) &0.85 (0.10) &0.84 (0.08) &1.68 (0.40) \\
Prop    &2.39 (0.33) &0.80 (0.11) &0.75 (0.12) &1.33 (0.17) &&1.61 (0.24) &0.80 (0.09) &0.73 (0.08) &2.55 (0.39) \\
\hline
\end{tabular}}
\end{center}
\end{table}

Table \ref{table:APE} reports the values of APE$_j$ for $j = 141, \ldots, 148$, corresponding to the prediction measurements for weeks 36 and 37.
The columns ``Hospitalizations" and ``PCR" represent the variables that the number of people hospitalized, and the number of people attending the PCR tests, respectively.
We observe that the Prop outperforms the SCIO and Glasso, and is slightly better than the CSCS and MCD approaches.
Specifically, the Prop provides more accurate prediction in the number of cases, which is of practically importance for pandemic study on the risk assessment.
It also predicts well in the number of deaths, and performs comparably with the CSCS and MCD in predicting the number of people who will take the PCR tests and the number of patients in hospital due to Covid-19.
Additionally, it is seen that the CSCS and MCD methods perform better than the Glasso and SCIO on this data set.
One possible reason is that the data, although not very strictly, have a natural time ordering by weeks among the variables.
The CSCS and MCD methods utilize this information and show the advantages over the Glasso and SCIO.

\begin{figure}[h]
\centering
\scalebox{0.28}[0.28]{\includegraphics{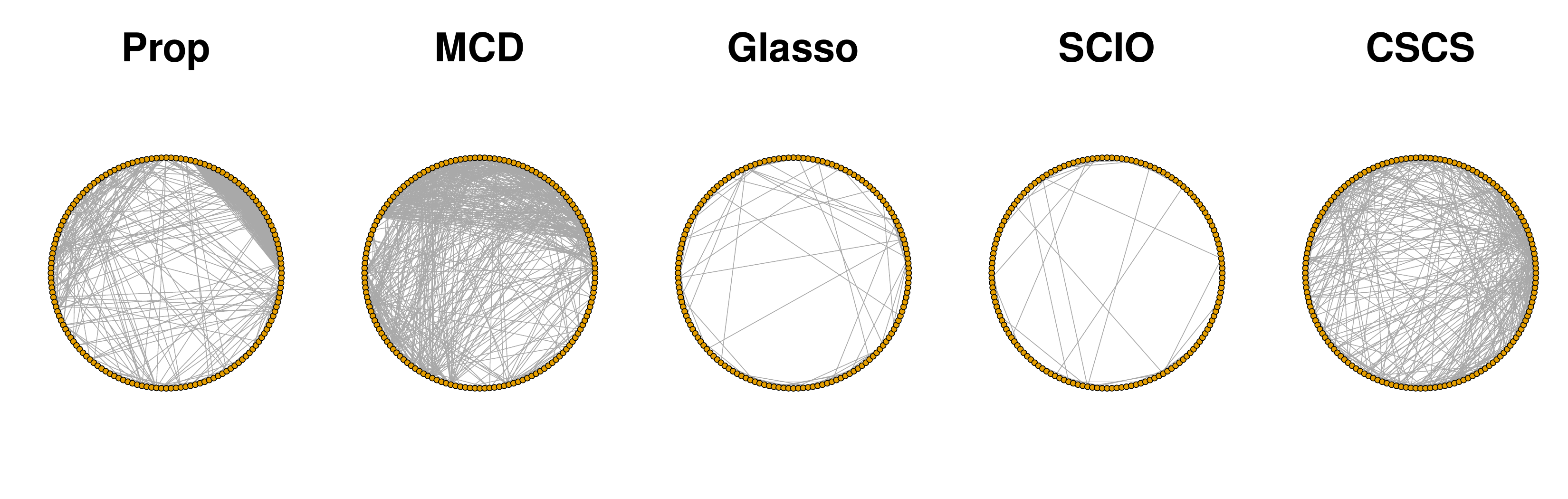}}
\caption{Conditional dependence networks inferred from Covid-19 data for all variables.}\label{Fig:networks}
\end{figure}

\begin{figure}[h]
\centering
\scalebox{0.28}[0.28]{\includegraphics{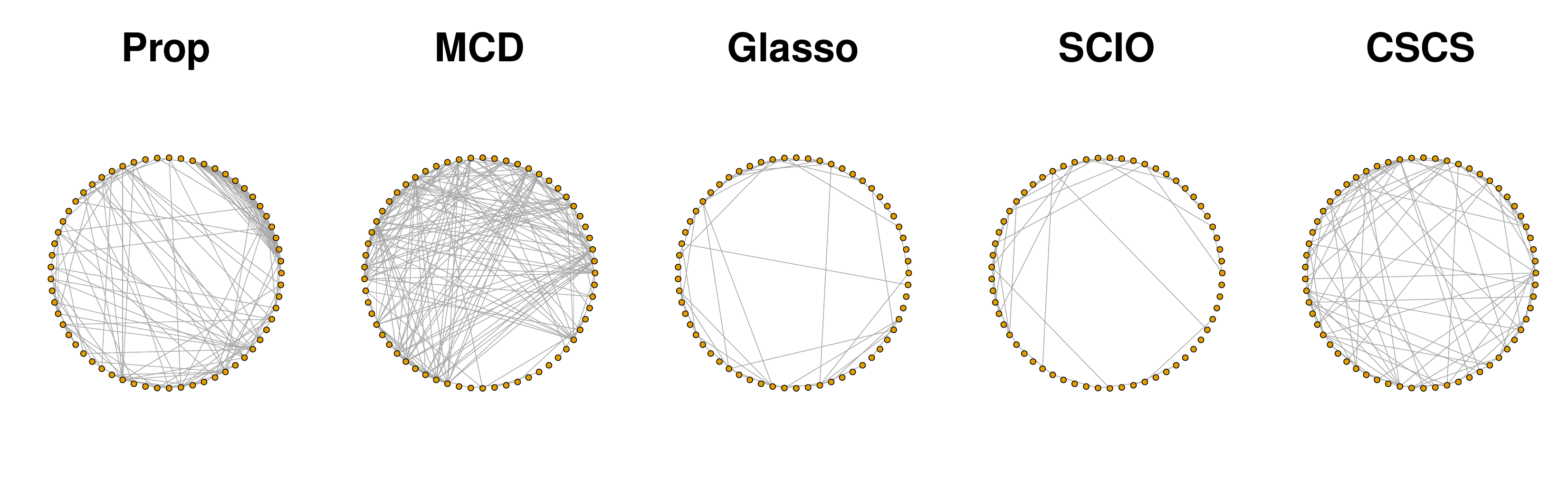}}
\caption{Conditional dependence networks inferred from Covid-19 data for variables from 5th week to 20th week.}\label{Fig:fewnetworks}
\end{figure}

In addition, Figure \ref{Fig:networks} displays the estimated conditional dependency relationship between variables through network plots, which provide a further insight into the analysis results from the compared methods.
It is seen that the CSCS and MCD methods yield an excessive amount of connections, which complicates the estimated model with difficult interpretation.
While the graphs constructed from the Glasso and SCIO appear to give few connections, not providing sufficient information to infer the conditional dependency among variables.
In contrast, the proposed model identifies a proper number of variable connections with certain sparsity.
For a clear presentation, Figure \ref{Fig:fewnetworks} shows the estimated networks for the 60 consecutive variables, corresponding to the 5th week to 20th week.
One would expect a meaningful network with nearby variables having some connections while the far-away variables having few connections.
From Figure \ref{Fig:fewnetworks}, it appears that the proposed model outperforms other methods, since the MCD method presents relatively too many connections of far-away variables, and the Glasso and SCIO provide too little information of dependency relationship between nearby variables.
The CSCS displays a proper number of connections as the Prop method, while the Prop method seems to infer more connections of nearby variables, which is indicated by dense connections on the both upper right and bottom right of network plot constructed by the Prop method.

\subsection{Call Center Data}
In this section, we apply the proposed model to analyze the call center data from \cite{huang2006covariance}.
The data set was collected from one of call centers in a major U.S. northeastern financial organization.
It recorded the time that every call arrives at the service queue from 7:00am until midnight in each day of 2002, except for 6 days when the data collecting
equipment was out of order.
The 17-hour period is divided into 102 10-minute intervals, and the number of calls during each time interval is counted.
Since the arrival patterns of calls are different between weekdays and weekends, the analysis focuses on weekdays here.
In addition, screening out some outliers that include holidays and days when the recording equipment was faulty, we are left with 239 observations.

Denote such data by $\bm N_i = (N_{i1},\ldots,N_{i102})', i = 1, . . . , 239$, where $N_{i j}$ is the number of calls received by the call center for the
$j$th 10-minute interval on day $i$.
Do the transformation $y_{ij} = \sqrt{(N_{ij} + 1/4)}, i = 1,\ldots, 239, j = 1,\ldots, 102$ to make the data distribution close to normal \citep{brown2005statistical}.
We apply the compared methods to estimate the $102 \times 102$ inverse covariance matrix.
The 102 variables are partitioned into three groups with each group subsequently containing 30, 30 and 42 variables, corresponding to three time intervals of 7:00am - 12:00am, 12:00am - 5:00pm and 5:00pm - midnight.
To compare the performance of different methods, we predict the number of arriving calls later in a day based on the arrival patterns at earlier times of that day.
As in Section \ref{sec:covid}, we use notation $\bm y_i = (y_{i1},\ldots,y_{i102})' = (\bm y_{iE}', \bm y_{iL}')'$, where $\bm y_{iE}$ and $\bm y_{iL}$ measure the arrival patterns in the early and later times of day $i$.
To examine the performance of the proposed model,
the 239 observations are split into a training set which contains the first 205 data corresponding to dates from January to October, and a testing set with the rest 34 data corresponding to dates from November and December.
The values of $\bm y_{iE}, i = 206, \ldots, 239$ in the testing data are used to predict $\bm y_{iL}$ based on Equation \eqref{eq:pred}.
Two different settings of $\bm y_{iE}$ are considered as (A): $\bm y_{iE} = (y_{i1}, \ldots, y_{i51})'$ and (B): $\bm y_{iE} = (y_{i1}, \ldots, y_{i60})'$.
Setting A is used in \cite{huang2006covariance} and represents using the data from the early half of a day to predict the call numbers in the later half of the day.
Setting B means using the data from the daytime to predict the call numbers in the evenings.
For each time interval in $\bm y_{iL}$, the APE is computed as
\begin{align*}
\mbox{APE}_j = \frac{1}{34} \sum_{i=206}^{239}|\hat{y}_{ij} - y_{ij}|,
\end{align*}
where $\hat{y}_{ij}$ is the predicted value, $j = 52, \ldots, 102$ for setting A and $j = 61, \ldots, 102$ for setting B.

\begin{table}[h]
\begin{center}
\caption{The averages and standard errors of APE for call center data.} \label{table:callcenter}
\begin{tabular}{rrrrrrrrrrrrrrrrr}
\hline
Setting &MCD     &Glasso    &SCIO      &CSCS  &Prop   \\
\hline
A       &1.240 (0.052) &1.258 (0.053) &1.194 (0.063) &1.152 (0.045) &0.991 (0.031) \\
B       &1.159 (0.045) &1.163 (0.044) &1.083 (0.030) &1.098 (0.039) &1.023 (0.030) \\
\hline
\end{tabular}
\end{center}
\end{table}

Table \ref{table:callcenter} shows the averages and corresponding standard errors of APE over $j$ from each compared method for Settings A and B.
The proposed model generally gives superior performance over other methods with lowest values of APE.
The SCIO appears to be comparable with the CSCS method, and they perform better than MCD and Glasso.
The application results demonstrate the merits of the proposed model when data have
partial information of the variable ordination.



\section{Discussion} \label{sec:con}
In this work, we propose a block Cholesky decomposition (BCD) method for inverse covariance estimation when the partial information of variable ordination is available.
The proposed model adopts the MCD to induce the sparsity by imposing the regularization on the multivariate regressions.
The proposed method provides a unified framework for several existing methods.
including the MCD, the Glasso, and the estimators of \cite{Rothman2010A} and \cite{Witten2011new}.
The theoretical results indicate that the proposed model enjoys a faster consistent rate than the Glasso if the data have the partial ordination among the variables.

There are several directions for future research.
First, the objective function \eqref{eq:obj} is not jointly convex, hence the theoretical global minimizer is not guaranteed.
A potential way to address this issue is to employ the classical Cholesky decomposition for the inverse covariance estimation \citep{yu2017learning, khare2019scalable}, which leads to a convex optimization.
However, the statistical interpretation of the classical Cholesky decomposition may not be as explicit as the MCD, where the Cholesky factor matrix can be constructed from regressions.
It will be interesting to further investigate how to incorporate the partial information of variable ordering into the classical Cholesky decomposition.
Second, the proposed method can be extended to investigate the multivariate response regression, where the multivariate response variables have certain ordering information, such as data with outcomes collected in a time sequence order.
Third, the error bounds in Theorem 1 and Theorem 2 could be further improved.
Note that $(\bm D_{j}^{-1})_\ast$ and $(\bm A_{j})_\ast$ in Theorem 1 are generated on the path of Algorithm 1.
When they are close to the underlying true parameters, several terms of $o_p(1)$ in the proof could be close to 0,
making a smaller constant term in the error bound.

\bibliographystyle{ECA_jasa}
\bibliography{Ref_BCD}

\section*{Appendix}
In this section, we provide all the technical proofs for the main results of the paper.
Before proving theorems, we present several lemmas.

\begin{lemma}\label{lemma3}
Denote a squared block diagonal matrix by $\bm D = \diag(\bm D_1, \bm D_2, \ldots,$ $\bm D_M)$. Suppose $\bm D_i$ have eigenvalues $\mathcal{C}_{\lambda_i} = \{ \lambda_{i p_1}, \lambda_{i p_2}, \ldots, \lambda_{i p_i} \}, i = 1, 2, \ldots, M$, then the eigenvalues of matrix $\bm D$ are $\mathcal{C}_{\lambda_1}, \mathcal{C}_{\lambda_2}, \ldots, \mathcal{C}_{\lambda_M}$.
\end{lemma}

\begin{proof}
Let $\bm D_i = \bm P_i \bm \Lambda_i \bm P_i^{-1}$ be the eigenvalue decomposition, where $\bm \Lambda_i = \diag(\lambda_{i p_1}, \lambda_{i p_2}, \ldots, \lambda_{i p_i})$, and $\bm P_i$ is composed of the corresponding eigenvectors.
Define $\bm \Lambda = \diag(\bm \Lambda_{1}, \bm \Lambda_{2}, \ldots, \bm \Lambda_{M})$ and $\bm P = \diag(\bm P_{1}, \bm P_{2}, \ldots, \bm P_{M})$.
Then we have
\begin{align*}
\bm D \bm P &= \diag(\bm D_1, \bm D_2, \ldots, \bm D_M) \diag(\bm P_{1}, \bm P_{2}, \ldots, \bm P_{M}) \\
&= \diag(\bm D_1 \bm P_{1} , \bm D_2 \bm P_{2} , \ldots, \bm D_M \bm P_{M}) \\
&= \diag(\bm P_{1} \bm \Lambda_{1}, \bm P_{2} \bm \Lambda_{2}, \ldots, \bm P_{M} \bm \Lambda_{M}) \\
&= \diag(\bm P_{1}, \bm P_{2}, \ldots, \bm P_{M}) \diag(\bm \Lambda_{1}, \bm \Lambda_{2}, \ldots, \bm \Lambda_{M}) \\
&= \bm P \bm \Lambda,
\end{align*}
which indicates $\bm D = \bm P \bm \Lambda \bm P^{-1}$, and establishes the lemma.
\end{proof}

Lemma \ref{lemma3} describes a property of eigenvalues for the block diagonal matrix.
The following Lemma \ref{lemma1} is from Theorem A.10 in \cite{Bai2010Spectral}.
It demonstrates the property of matrix singular values. Its result is stated here for completeness.

\begin{lemma}\label{lemma1}
Let $\bm B$ and $\bm C$ be two matrices of order $m_1 \times m_2$ and $m_2 \times m_3$. For any $i, j \geq 0$, we have
\begin{align*}
\varphi_{i+j+1} (\bm B \bm C) \leq \varphi_{i+1} (\bm B)  \varphi_{j+1} (\bm C).
\end{align*}
\end{lemma}

Based on the results of Lemmas \ref{lemma3} and \ref{lemma1}, we present the following Lemma \ref{lemma2}, which provides a relationship between matrix $\bm \Omega$ and its block Cholesky factor matrices ($\bm T^{-1}, \bm D^{-1}$) in terms of their singular values.

\begin{lemma}\label{lemma2}
Let $\bm \Omega = \bm T' \bm D^{-1} \bm T$ be the block MCD of the inverse covariance matrix.
If the condition \eqref{assumption1} is satisfied, that is, there exists a constant $\theta > 0$ such that
$1/\theta < \varphi_{p}(\bm \Omega) \leq \varphi_{1}(\bm \Omega) < \theta$,
then there exist constants $h_{1}$ and $h_{2}$ such that
\begin{align*}
0 < h_{1} < \varphi_{p}(\bm T^{-1}) \leq \varphi_{1}(\bm T^{-1}) < h_{2}< \infty,
\end{align*}
and
\begin{align*}
0 < h_{1} < \varphi_{p}(\bm D^{-1}) \leq \varphi_{1}(\bm D^{-1}) < h_{2}< \infty.
\end{align*}
\end{lemma}

\begin{proof}
By the decomposition \eqref{inverse covariance:eq1}, we partition $\bm \Omega$ into blocks according to the variable groups $\bm X^{(1)}, \bm X^{(2)}, \ldots, \bm X^{(M)}$ such that its diagonal blocks are $\bm \Omega_{ii}$ of order $p_i \times p_i, i = 1, 2, \ldots, M$, and $\sum_{i = 1}^M p_i = p$.
Write $\bm \Omega = \bm T' \bm D^{-1} \bm T = \bm T' \bm D^{-\frac{1}{2}} \bm D^{-\frac{1}{2}} \bm T = \bm R' \bm R$, where
\begin{align*}
\bm R = \bm D^{-\frac{1}{2}} \bm T = \left(
\begin{array}{ccccc}
\bm R_{11} & \bm 0 & \ldots & \bm 0 \\
\bm R_{21} &  \bm R_{22} & \ldots & \bm 0 \\
\vdots & \vdots & \ddots &  \vdots \\
\bm R_{M1} & \bm R_{M2} & \ldots & \bm R_{MM}
\end{array}
\right )
\end{align*}
with $\bm R_{ii} = \bm D_{i}^{-\frac{1}{2}}$.
Note that $\bm R_{ii}$ is a symmetric matrix due to the symmetry of $\bm D_{i}$.
In addition, it is obvious to have $\bm \Omega_{ii} = \sum_{i \geq k} \bm R'_{ik} \bm R_{ik}$, implying that $\bm \Omega_{ii} - \bm R'_{ii} \bm R_{ii} = \bm \Omega_{ii} - \bm D_i^{-1}$ is semi-positive definite.
Consequently we have
\begin{align}\label{proof;eq1}
\varphi_p(\bm D_i^{-1}) \leq \varphi_1(\bm D_i^{-1}) \leq \varphi_1(\bm \Omega_{ii}) \leq \theta.
\end{align}
Taking determinant on both sides of $\bm \Omega = \bm T' \bm D^{-1} \bm T$ yields
$$\varphi_p(\bm \Omega) \cdots \varphi_1(\bm \Omega) = \varphi_p(\bm D^{-1}) \cdots \varphi_1(\bm D^{-1}).$$
By $\varphi_1(\bm D_i^{-1}) \leq \theta$ for each $i = 1, 2, \ldots, M$ via \eqref{proof;eq1}, together with Lemma \ref{lemma3},
it is easy to see $\varphi_1(\bm D^{-1}) \leq \theta$. We hence have
\begin{align*}
(\frac{1}{\theta})^p \leq \varphi_p^{p}(\bm \Omega) \leq \prod_{i=1}^p \varphi_i(\bm \Omega) = \prod_{i=1}^p \varphi_i(\bm D^{-1}) \leq \theta^{p-1} \varphi_p(\bm D^{-1}),
\end{align*}
which gives $\varphi_p(\bm D^{-1}) \geq (\frac{1}{\theta})^{2p-1}$.
As a result,
\begin{align*}
0 < (\frac{1}{\theta})^{2p-1} \leq \varphi_p(\bm D^{-1}) \leq \varphi_1(\bm D^{-1}) \leq \theta < \infty.
\end{align*}
To bound singular values of matrix $\bm T^{-1}$, on one hand, we use Lemma \ref{lemma1} to obtain $\varphi_p(\bm \Omega) = \varphi_p(\bm T' \bm D^{-1} \bm T) = \varphi_p( \bm T \bm T' \bm D^{-1}) \leq \varphi_p( \bm T \bm T') \varphi_1(\bm D^{-1}) = \varphi_p(\bm T') \varphi_p(\bm T) \varphi_1(\bm D^{-1})$, indicating $$\varphi_p(\bm T) \geq \sqrt{\varphi_p(\bm \Omega) / \varphi_1(\bm D^{-1})} \geq \sqrt{1 / \theta^2} = \frac{1}{\theta}.$$
On the other hand, applying Lemma \ref{lemma1} again for $\bm D^{-1} = \bm T'^{-1} \bm \Omega \bm T^{-1}$ yields $\varphi_p (\bm D^{-1}) \leq  \varphi_p^2 (\bm T^{-1}) \varphi_1 (\bm \Omega) = \varphi_1 (\bm \Omega) / \varphi_1^2 (\bm T)$, implying
\begin{align*}
\varphi_1 (\bm T) \leq \sqrt{\frac{\varphi_1 (\bm \Omega)}{\varphi_p (\bm D^{-1})}} \leq \sqrt{\frac{\theta}{1 / \theta^{2p-1}}} = \theta^p.
\end{align*}
As a result,
\begin{align*}
0 < \frac{1}{\theta} \leq \varphi_p(\bm T) \leq \varphi_1(\bm T) \leq \theta^p < \infty \\
0 < (\frac{1}{\theta})^p \leq \varphi_p(\bm T^{-1}) \leq \varphi_1(\bm T^{-1}) \leq \theta < \infty.
\end{align*}
Taking $h_1 = \min(\theta^{1-2p}, \theta^{-p})$ and $h_2 = \theta$ establishes the lemma.
\end{proof}

It is seen from Lemma \ref{lemma2} that the singular values of the matrices $\bm T^{-1}$ and $\bm D^{-1}$ are bounded if the singular values of the inverse covariance matrix $\bm \Omega$ are bounded.
Now we give the proofs of Theorems.

\begin{proof}{\textbf{Proof of Theorem \ref{theorem1}}.}{}

From the negative log-likelihood \eqref{inverse covariance:eq5}, we have
\begin{align*}
L(\bm T,\bm D) &= -\sum_{j=1}^{M} \log |\bm D_{j}^{-1}| +  \sum_{j=1}^{M} \tr \left[ \bm S_{\epsilon_j} \bm D_{j}^{-1} \right ]  \nonumber \\
&= \sum_{j=1}^{M} \log |\bm D_{j}| +  \tr \left(
\begin{array}{ccccc}
\bm S_{\epsilon_1} \bm D_1^{-1} & \bm 0 & \ldots & \bm 0 \\
\bm 0 &  \bm S_{\epsilon_2} \bm D_2^{-1} & \ldots & \bm 0 \\
\vdots & \vdots & \ddots &  \vdots \\
\bm 0 & \bm 0 & \ldots & \bm S_{\epsilon_M} \bm D_M^{-1}
\end{array}
\right )   \nonumber \\
&= \log |\bm D| +  \tr \left(
\begin{array}{ccccc}
\bm S_{\epsilon_1} & \bm 0 & \ldots & \bm 0 \\
\bm 0 &  \bm S_{\epsilon_2} & \ldots & \bm 0 \\
\vdots & \vdots & \ddots &  \vdots \\
\bm 0 & \bm 0 & \ldots & \bm S_{\epsilon_M}
\end{array}
\right ) \bm D^{-1}.
\end{align*}
By the notation $\bm S_{\epsilon_j} = \frac{1}{n}(\mathbb{X}^{(j)} - \mathbb{Z}^{(j)} \bm A'_{j})'(\mathbb{X}^{(j)} - \mathbb{Z}^{(j)} \bm A'_{j})$, it is easy to see
\begin{align*}
L(\bm T,\bm D) &= \log |\bm D| +  \frac{1}{n} \tr \left(
\begin{array}{ccccc}
(\mathbb{X}^{(1)})' \\
(\mathbb{X}^{(2)} - \mathbb{Z}^{(2)} \bm A'_{2} )'  \\
\vdots \\
(\mathbb{X}^{(M)} - \mathbb{Z}^{(M)} \bm A'_{M} )'
\end{array}
\right )  \\
&~~~~~~~~~~~~~~~~~\left( \mathbb{X}^{(1)},  \mathbb{X}^{(2)} - \mathbb{Z}^{(2)} \bm A'_{2}, \ldots, \mathbb{X}^{(M)} - \mathbb{Z}^{(M)} \bm A'_{M}   \right ) \bm D^{-1}    \\
&= \log |\bm D| + \frac{1}{n} \tr \left[ \bm T \mathbb{X}' \mathbb{X} \bm T' \bm D^{-1} \right ]   \\
&= \log |\bm D| + \tr \left[ \bm T' \bm D^{-1} \bm T \bm S \right ],
\end{align*}
where $\bm S = \frac{1}{n} \mathbb{X}' \mathbb{X}$.
Consequently, $L_{\lambda}(\bm T,\bm D)$ can be written as
\begin{align*}
L_{\lambda}(\bm T,\bm D) &= \log |\bm D| + \tr \left[ \bm T' \bm D^{-1} \bm T \bm S \right ] + \lambda_{1} \| \bm A \|_{1} +  \lambda_{2}  \| \bm D^{-1} \|_{1}^{-}  \\
&= \log |\bm D| + \tr \left[ \bm T' \bm D^{-1} \bm T \bm S \right ] + \lambda_{1} \| \bm T \|_{1} +  \lambda_{2}  \| \bm D^{-1} \|_{1}^{-}  \\
&= \log |\bm D| + \tr \left[ \bm T' \bm D^{-1} \bm T \bm S \right ] + \lambda_{1} \sum_{i > k} |t_{ik}| +  \lambda_{2} \sum_{i \neq k} |\psi_{ik}|,
\end{align*}
where $t_{ik}$ and $\psi_{ik}$ are the $(i, k)$th elements of matrices $\bm T$ and $\bm D^{-1}$, respectively.

For part (a), we define $G_1(\Delta_{T}) = L_{\lambda}(\bm T_{0} + \Delta_{T} | \bm D_\ast) - L_{\lambda}(\bm T_{0} | \bm D_\ast)$.
Let $\mathcal{A}_{U_{j}} = \{ \Delta_{T_j}: \Delta_{T_j} = \Delta_{T_j}', \| \Delta_{T_j} \|_{F}^{2} \leq U_{j}^{2} s_{T_j} \log(\sum_{k=1}^j p_k) / n \}$ for $j = 1, 2, \ldots, M$, where $U_{j}$ are positive constants.
We will show that for $\Delta_{T_j} \in \partial \mathcal{A}_{U_{j}}$, probability $\Pr(G_1(\Delta_{T})) > 0$ is tending to 1 as $n \rightarrow \infty$ for sufficiently large $U_{j}$, where $\partial \mathcal{A}_{U_{j}}$ are the boundaries of $\mathcal{A}_{U_{j}}$.
Additionally, since $G_1(\Delta_{T}) = 0$ when $\Delta_{T_j} = 0$, the minimum point of $G_1(\Delta_{T})$ is achieved when
$\Delta_{T_j} \in \mathcal{A}_{U_{j}}$.
That is $\| \Delta_{T_j} \|_{F}^{2} = O_{p} (s_{T_j} \log(\sum_{k=1}^j p_k) / n)$.

Assume $\| \Delta_{T_j} \|_{F}^{2} = U_{j}^{2} s_{T_j} \log(\sum_{k=1}^j p_k) / n$.
Write $\bm T = \bm T_{0} + \Delta_{T}$,
then we decompose $G_1(\Delta_{T})$ as
\begin{align*}
G_1(\Delta_{T}) &= L_{\lambda}(\bm T_{0} + \Delta_{T} | \bm D_\ast) - L_{\lambda}(\bm T_{0} | \bm D_\ast) \\
&= \tr \left[ \bm T' \bm D_\ast^{-1} \bm T \bm S \right ] - \tr \left[ \bm T_0' \bm D_\ast^{-1} \bm T_0 \bm S \right ] + \lambda_{1} \sum |t_{ik}| - \lambda_{1} \sum |t_{0ik}|  \\
&= M_{1} + M_{2} + M_{3},
\end{align*}
where
\begin{align*}
M_{1} &= \tr [\bm D_\ast^{-1} (\bm T (\bm S - \bm \Sigma_{0}) \bm T' - \bm T_{0} (\bm S - \bm \Sigma_{0}) \bm T'_{0})], \\
M_{2} &= \tr [\bm D_\ast^{-1} (\bm T \bm \Sigma_{0} \bm T' - \bm T_{0} \bm \Sigma_{0} \bm T'_{0})], \\
M_{3} &= \lambda_{1} \sum |t_{ik}| -  \lambda_{1} \sum |t_{0ik}|.
\end{align*}
The above decomposition of $G_1(\Delta_{T})$ into $M_{1}$ to $M_{3}$ is very similar to that in the proof of Lemma 3 of \cite{kang2021on}; hence it is omitted here.
Now we bound each component respectively.
Note that $\|\Delta_{T}\|_{F}^{2} = \| \bm T - \bm T_0 \|_{F}^{2} = \sum_{j=1}^M \|\Delta_{T_j} \|_{F}^{2}$.
Therefore, based on the proof of Theorem 3.1 in \cite{Jiang2012Joint}, for any $\epsilon > 0$, there exists a constant $V_{1} > 0$ such that with probability greater than $1 - \epsilon$, we have
\begin{align*}
&M_{2} - |M_{1}| \\
&> \frac{ \|\Delta_{T} \|_{F}^{2} }{h^4} -
V_{1} \sum_{j=1}^M \left( || \bm T_j - \bm T_{j0} ||_1 \sqrt{\log (\sum_{k=1}^j p_k) / n}  \right)     \\
&= \frac{ \sum_{j=1}^M \|\Delta_{T_j} \|_{F}^{2} }{h^4} - V_{1} \sum_{j=1}^M \left( \sqrt{\log (\sum_{k=1}^j p_k) / n} \sum_{(i, k) \in \mathcal{Z}_{T_j}^c}|t_{ik}| \right) \\
&~~~~~~~~~~~~~~~~~~~~~~~~- V_{1} \sum_{j=1}^M \left( \sqrt{\log (\sum_{k=1}^j p_k) / n} \sum_{(i, k) \in \mathcal{Z}_{T_j}}|t_{ik} - t_{0ik}| \right)     \\
&\geq \frac{ \sum_{j=1}^M \|\Delta_{T_j} \|_{F}^{2} }{h^4}- V_{1} \sqrt{\log (p) / n} \sum_{(i, k) \in \bigcup_{j=1}^M \mathcal{Z}_{T_j}^c}|t_{ik}| - V_{1} \sum_{j=1}^M  \sqrt{ s_{T_j} \log (\sum_{k=1}^j p_k) / n  \|\Delta_{T_j}\|_{F}^2 } \\
&= \frac{1}{h^4} \sum_{j=1}^M U_j^2 s_{T_j} \log (\sum_{k=1}^j p_k) / n - V_{1} \sqrt{\log (p) / n} \sum_{(i, k) \in \bigcup_{j=1}^M \mathcal{Z}_{T_j}^c}|t_{ik}| \\
&~~~~~~~~~~~~~~~~~~~~~~~~ - V_{1} \sum_{j=1}^M U_j s_{T_j} \log (\sum_{k=1}^j p_k) / n  \\
&\geq \frac{1}{n h^4} \sum_{j=1}^M U_j^2 s_{T_j} (\log \gamma_j + C_j \log p) - V_{1} \sqrt{\log (p) / n} \sum_{(i, k) \in \bigcup_{j=1}^M \mathcal{Z}_{T_j}^c}|t_{ik}| - V_{1} \frac{\log p}{n} \sum_{j=1}^M U_j s_{T_j}  \\
&\geq \frac{1}{n h^4}\sum_{j=1}^M U_j^2 s_{T_j} \log \gamma_j + \frac{1}{h^4} \frac{\log p}{n} \tau_c \tau_u \sum_{j=1}^M U_j s_{T_j}- V_{1} \sqrt{\log (p) / n} \sum_{(i, k) \in \bigcup_{j=1}^M \mathcal{Z}_{T_j}^c}|t_{ik}|  \\
&~~~~~~~~~~~~~~~~~~~~~~~~- V_{1} \frac{\log p}{n} \sum_{j=1}^M U_j s_{T_j},
\end{align*}
where $\tau_u$ is a positive constant satisfying $\tau_u \leq U_{j}, j = 1, 2, \ldots, M$.
Next, for the penalty term corresponding to $\lambda_1$,
\begin{align*}
M_{3} = \lambda_1 \sum_{(i, k) \in \bigcup_{j=1}^M \mathcal{Z}_{T_j}^c}|t_{ik}| + \lambda_1
\sum_{(i, k) \in \bigcup_{j=1}^M \mathcal{Z}_{T_j}} (|t_{ik}| - |t_{0ik}|) = M_{3}^{(1)} + M_{3}^{(2)},
\end{align*}
where $M_{3}^{(1)} = \lambda_1 \sum_{(i, k) \in \bigcup_{j=1}^M \mathcal{Z}_{T_j}^c}|t_{ik}|$,
and
\begin{align*}
| M_{3}^{(2)} | = | \lambda_1 \sum_{(i, k) \in \bigcup_{j=1}^M \mathcal{Z}_{T_j}}(|t_{ik}| - |t_{0ik}|) |
&\leq \lambda_1 \sum_{(i, k) \in \bigcup_{j=1}^M \mathcal{Z}_{T_j}}| t_{ik} - t_{0ik} |   \\
&\leq \lambda_1  \sum_{j=1}^M \|\Delta_{T_j} \|_{F} \sqrt{s_{T_j}}  \\
&\leq \lambda_1 \sqrt{\frac{\log p}{n}} \sum_{j=1}^M U_j s_{T_j}.
\end{align*}
Combine all the terms above together, with probability greater than $1 - \epsilon$, we have
\begin{align*}
&~~~ G_1(\Delta_{T}) \geq M_{2} - |M_{1}| + M_{3}^{(1)} - |M_{3}^{(2)}|   \\
&\geq \frac{1}{nh^4} \sum_{j=1}^M U_j^2 s_{T_j} \log \gamma_j + \frac{1}{h^4} \frac{\log p}{n} \tau_c \tau_u \sum_{j=1}^M U_j s_{T_j} - V_{1} \sqrt{\log (p) / n} \sum_{(i, k) \in \bigcup_{j=1}^M \mathcal{Z}_{T_j}^c}|t_{ik}|  \\
&~~~  - V_{1} \frac{\log p}{n} \sum_{j=1}^M U_j s_{T_j} + \lambda_1 \sum_{(i, k) \in \bigcup_{j=1}^M \mathcal{Z}_{T_j}^c}|t_{ik}| - \lambda_1 \sqrt{\frac{\log p}{n}} \sum_{j=1}^M U_j s_{T_j} \\
&= \frac{1}{nh^4} \sum_{j=1}^M U_j^2 s_{T_j} \log \gamma_j + \frac{\log p \sum_{j=1}^M U_{j} s_{T_j}}{n} (\frac{\tau_c \tau_u}{h^{4}} - V_{1} - \frac{\lambda_1}{\sqrt{\log (p) / n}}) \\
&~~~ + (\lambda_1 - V_{1} \sqrt{\log (p) / n}) \sum_{(i, k) \in \bigcup_{j=1}^M \mathcal{Z}_{T_j}^c}|t_{ik}|.
\end{align*}
Here $V_{1}$ is only related to the sample size $n$ and $\epsilon$.
Assume $\lambda_1 = K_1 \sqrt{\log (p) / n}$ where $K_1 > V_{1}$, and choose $\tau_u > h^4 (K_1 + V_{1}) / \tau_c$, then $G_1(\Delta_{T}) > 0$.
Therefore, we prove $\| \Delta_{T_j} \|_{F}^{2} = O_{p} (s_{T_j} \log(\sum_{k=1}^j p_k) / n)$.

The proof of part (b) follows the same principle as that for part (a).
Similarly, define $G_2(\Delta_{D}) = L_{\lambda}(\bm D_{0} + \Delta_{D} | \bm T_\ast) - L_{\lambda}(\bm D_{0} | \bm T_\ast)$.
Let $\mathcal{B}_{W_j} = \{ \Delta_{D_j}: \Delta_{D_j} = \Delta_{D_j}', \| \Delta_{D_j} \|_{F}^{2} \leq W_{j}^{2} (s_{D_j} + p_j) \log(p_j) / n \}$ for $j = 1, 2, \ldots, M$, where $W_{j}$ are positive constants.
We only need to show that for $\Delta_{D_j} \in \partial \mathcal{B}_{W_{j}}$, probability $P(G_2(\Delta_{D}) > 0)$ is tending to 1 as $n \rightarrow \infty$ for sufficiently large $W_{j}$, where $\partial \mathcal{B}_{W_{j}}$ are the boundaries of $\mathcal{B}_{W_{j}}$.

Assume $\| \Delta_{D_j} \|_{F}^{2} = W_{j}^{2} (s_{D_j} + p_j) \log(p_j) / n$.
Write $\bm D = \bm D_{0} + \Delta_{D}$, then we decompose $G_2(\Delta_{D})$ as
\begin{align*}
G_2&(\Delta_{D}) = L_{\lambda}(\bm D_{0} + \Delta_{D} | \bm T_\ast) - L_{\lambda}(\bm D_{0} | \bm T_\ast) \\
&= \log |\bm D| - \log |\bm D_0| + \tr \left[ \bm T_\ast' \bm D^{-1} \bm T_\ast \bm S - \bm T_\ast' \bm D_0^{-1} \bm T_\ast \bm S \right ] + \lambda_{2} \sum_{i \neq k} |\psi_{ik}| - \lambda_{2} \sum_{i \neq k}|\psi_{0ik}|  \\
&= M_{4} + M_{5} + M_{6},
\end{align*}
where
\begin{align*}
M_{4} &= \log \left| \bm D \right| - \log \left| \bm D_{0} \right| + \tr[(\bm D^{-1} - \bm D_{0}^{-1})\bm D_{0}],\\
M_{5} &= \tr(\bm D^{-1} - \bm D_{0}^{-1})[\bm T_\ast (\bm S - \bm \Sigma_{0}) \bm T_\ast'],\\
M_{6} &= \lambda_{2} \sum_{i \neq k} |\psi_{ik}| -  \lambda_{2} \sum_{i \neq k} |\psi_{0ik}|.
\end{align*}
The above decomposition of $G_2(\Delta_{D})$ into $M_{4}$ to $M_{6}$ is similar to that in the proof of Lemma 3 of \cite{kang2021on}.
Next we bound each component respectively.
Note that $\|\Delta_{D}\|_{F}^{2} = \| \bm D - \bm D_0 \|_{F}^{2} = \sum_{j=1}^M \| \bm D_j - \bm D_{j0} \|_{F}^{2} = \sum_{j=1}^M \|\Delta_{D_j} \|_{F}^{2}$.
Therefore, based on the proof of Theorem 3.1 in \cite{Jiang2012Joint} together with Lemma \ref{lemma2}, we can have the following two results (I) and (II).

(I) Let $\tau_w$ be a positive constant satisfying $\tau_w \leq W_{j}, j = 1, 2, \ldots, M$, and note that $(1/h^2) \|\Delta_{D}\|_{F}^{2} \leq \| \bm D^{-1} - \bm D_0^{-1} \|_{F}^{2} \leq h^2 \|\Delta_{D}\|_{F}^{2}$, then
\begin{align*}
M_{4} &\geq \frac{1}{8h^2} \| \bm D^{-1} - \bm D_0^{-1} \|_{F}^{2} \geq \frac{1}{8h^4} \|\Delta_{D}\|_{F}^{2}
= \frac{1}{8h^4} \sum_{j=1}^M \|\Delta_{D_j} \|_{F}^{2} \\
&= \frac{1}{8h^4} \sum_{j=1}^M W_{j}^{2} (s_{D_j} + p_j) \log(p_j) / n  \\
&= \frac{1}{8h^4} \sum_{j=1}^M W_{j}^{2} (s_{D_j} + p_j) (\log \gamma_j + C_j \log p) / n  \\
&\geq \frac{1}{8nh^4} \sum_{j=1}^M W_{j}^2 (s_{D_j} + p_j) \log \gamma_j + \frac{1}{8h^4} \frac{\log p}{n} \tau_c \tau_w \sum_{j=1}^M W_{j} (s_{D_j} + p_j).
\end{align*}

(II) For any $\epsilon > 0$, there exists a constant $V_{2} > 0$ such that with probability greater than $1 - \epsilon$, we have
\begin{align*}
|M_{5}| = |\tr(\bm D^{-1} - \bm D_{0}^{-1})[\bm T_\ast (\bm S - \bm \Sigma_{0}) \bm T_\ast']| &\leq \max |\xi_{ik}| \sum_{j=1}^M || \bm D_j - \bm D_{j0} ||_1 \\
&\leq V_{2} \sqrt{\frac{\log p}{n}} \sum_{j=1}^M \sqrt{ (s_{D_j} + p_j) \|\Delta_{D_j}\|_{F}^2 } \\
&\leq V_{2} \frac{\log p}{n} \sum_{j=1}^M W_j (s_{D_j} + p_j),
\end{align*}
where $\xi_{ik}$ is the $(i,k)$th element of matrix $\bm T_\ast (\bm S - \bm \Sigma_{0}) \bm T_\ast'$, and the second inequality applies Lemma 3 of \cite{lam2009sparsistency}.

Next, we decompose $M_{6} = M_{6}^{(1)} + M_{6}^{(2)}$, where $M_{6}^{(1)} =  \lambda_2 \sum_{(i, k) \in \bigcup_{j=1}^M \mathcal{Z}_{D_j}^c}|\psi_{ik}|$,
and
\begin{align*}
| M_{6}^{(2)} | \leq  | \lambda_2 \sum_{(i, k) \in \bigcup_{j=1}^M \mathcal{Z}_{D_j}}(|\psi_{ik}| - |\psi_{0ik}|) |
&\leq  \lambda_2 \sum_{(i, k) \in \bigcup_{j=1}^M \mathcal{Z}_{D_j}}| \psi_{ik} - \psi_{0ik} |   \\
&\leq  \lambda_2  \sum_{j=1}^M \|\Delta_{D_j} \|_{F} \sqrt{s_{D_j} + p_j} \\
&\leq  \lambda_2 \sqrt{\frac{\log p}{n}} \sum_{j=1}^M W_j (s_{D_j} + p_j).
\end{align*}
Combine all the terms above together, with probability greater than $1 - \epsilon$, we have
\begin{align*}
&~~~ G_2(\Delta_{D}) \geq M_{4} - |M_{5}| + M_{6}^{(1)} - |M_{6}^{(2)}|   \\
&\geq \frac{1}{8nh^4} \sum_{j=1}^M W_{j}^2 (s_{D_j} + p_j) \log \gamma_j + \frac{1}{8h^4} \frac{\log p}{n} \tau_c \tau_w \sum_{j=1}^M W_{j} (s_{D_j} + p_j) \\
&~~~ - V_{2} \frac{\log p}{n} \sum_{j=1}^M W_j (s_{D_j} + p_j) + \lambda_2 \sum_{(i, k) \in \bigcup_{j=1}^M \mathcal{Z}_{D_j}^c}|\psi_{ik}| - \lambda_2 \sqrt{\frac{\log p}{n}} \sum_{j=1}^M W_j (s_{D_j} + p_j)   \\
&= \frac{1}{8nh^4} \sum_{j=1}^M W_{j}^2 (s_{D_j} + p_j) \log \gamma_j + \frac{\log p \sum_{j=1}^M W_{j} (s_{D_j} + p_j)}{n} (\frac{\tau_c \tau_w}{8h^4} - V_{2} - \frac{\lambda_2}{\sqrt{\log (p) / n}}) \\
&~~~ +  \lambda_2 \sum_{(i, k) \in \bigcup_{j=1}^M \mathcal{Z}_{D_j}^c}|\psi_{ik}|.
\end{align*}
Here $V_{2}$ is only related to the sample size $n$ and $\epsilon$.
Assume $\lambda_2 = K_2 \sqrt{\log (p) / n}$ where $K_2 > 0$, and choose $\tau_w > 8h^4 (K_2 + V_{1}) / \tau_c$, then $G_2(\Delta_{D}) > 0$.
Therefore, we prove $\| \Delta_{D_j} \|_{F}^{2} = W_{j}^{2} (s_{D_j} + p_j) \log(p_j) / n$.
\end{proof}

\begin{proof}{\textbf{Proof of Theorem \ref{theorem2}}.}{}

Let $\hat{\bm A}_j$ and $\hat{\bm D}_j$ be the estimates obtained from Step 3 in Algorithm 1. We first prove the consistent rates under Frobenius norm of $\hat{\bm T}_j = -\hat{\bm A}_j$ and $\hat{\bm D}_j$ are $s_{T_j} \log(\sum_{k=1}^j p_k) / n$ and $(s_{D_j} + p_j) \log(p_j) / n$, respectively.

At the first iteration of Step 1 in Algorithm 1, the estimate $\hat{\bm A}_{j;1}$ found by minimizing $\ell_{\lambda}(\bm A_j | \bm I)$ is $s_{T_j} \log(\sum_{k=1}^j p_k) / n$ consistent according to part (a) of Theorem \ref{theorem1}.
In Step 2, the estimate $\hat{\bm D}_{j;1}$ found by minimizing $\ell_{\lambda}(\bm D_j^{-1} | \hat{\bm A}_{j;1})$ is $(s_{D_j} + p_j) \log(p_j) / n$ consistent according to part (b) of Theorem \ref{theorem1}.
Next, an estimate $\hat{\bm A}_{j;2}$ obtained by minimizing $\ell_{\lambda}(\bm A_j | \hat{\bm D}_{j;1})$ is $s_{T_j} \log(\sum_{k=1}^j p_k) / n$ consistent, and $\hat{\bm D}_{j;2}$ which minimizes $\ell_{\lambda}(\bm D_j^{-1} | \hat{\bm A}_{j;2})$ is $(s_{D_j} + p_j) \log(p_j) / n$ consistent.
Following this, we hence have that $\hat{\bm A}_j$ (or equivalently $\hat{\bm T}_j$) and $\hat{\bm D}_j$ are $s_{T_j} \log(\sum_{k=1}^j p_k) / n$ and $(s_{D_j} + p_j) \log(p_j) / n$ consistent.
This implies
\begin{align*}
\| \hat{\bm T} - \bm T_{0} \|_{F}^2 = \sum_{j=1}^M || \hat{\bm T}_j - \bm T_0 ||_{F}^2 = \sum_{j=1}^M O_p( s_{T_j} \log(\sum_{k=1}^j p_k) / n ) \leq O_p(s_{T} \log(p) / n)
\end{align*}
and
\begin{align*}
\| \hat{\bm D} - \bm D_{0} \|_{F}^2  = \sum_{j=1}^M O_p \left( (s_{D_j} + p_j) \log(p_j) / n \right) = O_p(\sum_{j=1}^M (s_{D_{j}} + p_{j}) \log (p_{j}) / n).
\end{align*}
Next, we derive of consistent rate of the estimate $\hat{\bm \Omega} = \hat{\bm T}'  \hat{\bm D}^{-1} \hat{\bm T}$.
Let $\bm \Delta_{T} = \hat{\bm T} - \bm T_{0}$ and $\bm \Delta_{D} = \hat{\bm D} - \bm D_{0}$, then we decompose $\| \hat{\bm \Omega} - \bm \Omega_{0} \|_{F}^2$ as
\begin{align*}
&\| \hat{\bm \Omega} - \bm \Omega_{0} \|_{F}^2 \\
&= \| \hat{\bm T}'  \hat{\bm D}^{-1} \hat{\bm T}  - \bm T_{0}' \bm D_{0}^{-1} \bm T_{0} \|_{F}^2 \\
&= \| (\bm \Delta_{T}' + \bm T'_{0}) \hat{\bm D}^{-1} (\bm \Delta_{T} + \bm T_{0}) - \bm T_{0}' \bm D_{0}^{-1} \bm T_{0} \|_{F}^2 \\
&\leq \| \bm \Delta_{T}' \hat{\bm D}^{-1} \bm T_{0} \|_{F}^2 + \| \bm T_{0}' \hat{\bm D}^{-1} \bm \Delta_{T} \|_{F}^2 + \| \bm \Delta_{T}' \hat{\bm D}^{-1} \bm \Delta_{T} \|_{F}^2 + \| \bm T_{0}' (\hat{\bm D}^{-1} - \bm D_{0}^{-1}) \bm T_{0} \|_{F}^2.
\end{align*}
Now we bound four terms separately.
Use the symbol $\| \bm A \|$ to represent the spectral norm of matrix $\bm A$.
Since $\| \bm T_0 \| = O(1)$ and $\| \bm D_0 \| = O(1)$ by Lemma \ref{lemma2}, it is obvious that
$\| \hat{\bm D} \| = \| \hat{\bm D} - \bm D_0 + \bm D_0\| \leq \| \bm \Delta_{D} \| + \| \bm D_0 \| \leq \| \bm \Delta_{D} \|_{F} + \| \bm D_0 \| = O_p(1)$.
In addition, the single values of $\bm \Omega^{-1}$ are bounded since the single values of $\bm \Omega$ are bounded, which together with Lemma \ref{lemma2} leads to
$\| \bm D_0^{-1} \| = O(1)$, hence similarly $\| \hat{\bm D}^{-1} \| = O_p(1)$.
As a result, it is easy to obtain
\begin{align*}
\| \bm \Delta_{T}' \hat{\bm D}^{-1} \bm T_{0} \|_{F}^2 \leq \| \bm \Delta_{T}'\|_{F}^2 \| \hat{\bm D}^{-1} \| \| \bm T_{0} \| = O_p(\| \bm \Delta_{T} \|_{F}^2),
\end{align*}
and the second term $\| \bm T_{0}' \hat{\bm D}^{-1} \bm \Delta_{T} \|_{F}^2 = \| \bm \Delta_{T}' \hat{\bm D}^{-1} \bm T_{0} \|_{F}^2 = O_p(\| \bm \Delta_{T} \|_{F}^2)$. For the third term,
\begin{align*}
\| \bm \Delta_{T}' \hat{\bm D}^{-1} \bm \Delta_{T} \|_{F}^2 \leq \| \bm \Delta_{T}' \|_{F}^2 \| \hat{\bm D}^{-1} \| \| \bm \Delta_{T} \|_{F}^2 = o_p({\| \bm \Delta_{T} \|_{F}^2}).
\end{align*}
For the fourth term,
\begin{align*}
\| \bm T_{0}' (\hat{\bm D}^{-1} - \bm D_{0}^{-1}) \bm T_{0} \|_{F}^2 \leq \| \bm T_{0}' \| \| \hat{\bm D}^{-1} - \bm D_{0}^{-1} \|_{F}^2 \| \bm T_{0} \| = O_p(\| \hat{\bm D} - \bm D_{0}\|_{F}^2).
\end{align*}
Consequently, by the convergence rates of $\hat{\bm T}  - \bm T_{0}$ and $\hat{\bm D} - \bm D_{0}$ from Theorem \ref{theorem1}, we reach the conclusion
\begin{align*}
\| \hat{\bm \Omega}  - \bm \Omega_{0} \|_{F}^2 &= O_{p}(\| \hat{\bm T}  - \bm T_{0} \|_{F}^2) +  O_{p}(\| \hat{\bm D} - \bm D_{0} \|_{F}^2) \\
&= O_p \left( \frac{s_{T} \log p + \sum_{j=1}^M (s_{D_{j}} + p_{j}) \log p_{j} }{n} \right).
\end{align*}
\end{proof}
\end{document}